\documentclass[12pt,reqno,intlimits]{amsart}

\usepackage{amssymb,amscd}

\newcommand{\nc}{\newcommand}

\nc{\al}{\alpha}
\nc{\bt}{\beta}
\nc{\gm}{\gamma}
\nc{\Gm}{\Gamma}
\nc{\dl}{\delta}
\nc{\Dl}{\Delta}
\nc{\lb}{\lambda}
\nc{\om}{\omega}
\nc{\Om}{\Omega}
\nc{\sg}{\sigma}
\nc{\tht}{\theta}
\nc{\Tht}{\Theta}
\nc{\vf}{\varphi}
\nc{\ve}{\varepsilon}
\nc{\zt}{\zeta}

\nc{\bF}{\mathbb{F}}
\nc{\bR}{\mathbb{R}}
\nc{\bL}{{\bf L}}
\nc{\bP}{{\bf P}}

\nc{\cA}{\mathcal{A}}
\nc{\cB}{\mathcal{B}}
\nc{\cE}{\mathcal{E}}
\nc{\cF}{\mathcal{F}}
\nc{\cH}{\mathcal{H}}
\nc{\cJ}{\mathcal{J}}
\nc{\cK}{\mathcal{K}}
\nc{\cM}{\mathcal{M}}
\nc{\cN}{\mathcal{N}}
\nc{\cL}{\mathcal{L}}
\nc{\cX}{\mathcal{X}}
\nc{\cP}{\mathcal{P}}
\nc{\cV}{\mathcal{V}}
\nc{\cv}{v}

\nc{\pa}{\partial}
\nc{\sbs}{\subset}
\nc{\sbseq}{\subseteq}
\nc{\wt}{\widetilde}
\nc{\wh}{\widehat}
\nc{\Ra}{\Rightarrow}
\nc{\Lra}{\Leftrightarrow}
\nc{\LLra}{\Longleftrightarrow}
\nc{\la}{\langle}
\nc{\ra}{\rangle}
\nc{\ua}{\uparrow}
\nc{\str}{\stackrel}
\nc{\ol}{\overline}
\nc{\ul}{\underline}
\nc{\os}{\overset}
\nc{\us}{\underset}
\nc{\fa}{\forall}
\nc{\ds}{\displaystyle}
\nc{\XA}{\{X\to\}}
\nc{\Pas}{P\text{-\it a.s.}}
\nc{\vnth}{\varnothing}

\nc{\loc}{\operatorname{loc}}
\nc{\rank}{\operatorname{rank}}
\nc{\tr}{\operatorname{tr}}
\nc{\diag}{\operatorname{diag}}
\nc{\Ball}{\operatorname{Ball}}
\newcommand{\const}{\operatorname{const}}

\newtheorem{theorem}{Theorem}[section]

\newtheorem{proposition}{Proposition}[section]

\theoremstyle{definition}
\newtheorem{definition}{Definition}[section]

\theoremstyle{remark}
\newtheorem{remark}{Remark}[section]

\numberwithin{equation}{section}


\begin{document}

\title[Optimal Robust Mean-Variance Hedging]{Optimal Robust Mean-Variance Hedging in Incomplete Financial Markets}

\author{N. Lazrieva, T. Toronjadze}

\maketitle

\begin{center}
Georgian--American University, Business School, 3 Alleyway II, \\Chavchavadze Ave. 17\,a, Tbilisi, Georgia, E-mail: toronj333@yahoo.com \\[2mm]
A. Razmadze Mathematical Institute, 1, M. Aleksidze St., Tbilisi, Georgia 
\end{center}

\begin{abstract}
Optimal $B$-robust  estimate is constructed for multidimensional parameter in drift coefficient of diffusion type process with small noise. Optimal mean-variance robust (optimal $V$-robust) trading strategy is find to hedge in mean-variance sense the contingent claim in incomplete financial market with arbitrary information structure and misspecified volatility of asset price, which is modelled by multidimensional continuous semimartingale. Obtained results are applied to stochastic volatility model, where the model of latent volatility process contains unknown multidimensional parameter in drift coefficient and small parameter in diffusion term.

\end{abstract}

%   1
\section{Introduction, Motivation and Results}

The hedging and pricing of contingent claims in incomplete financial markets, and dynamic portfolio selection problems are important issues in modern theory of finance. These problems are associated due to the so-called mean-variance approach. 

For determining a ``good'' hedging strategy in incomplete market with arbitrary information structure $F=(\cF)_{0\leq t\leq T}$, one riskless asset and $d$, $d\geq 1$, risky assets, whose price process is a semimartingale $X$, the mean-variance approach suggests to use the quadratic criterion to measure the hedging error, i.e. to solve the mean-variance hedging problem introduced by F\"{o}llmer and Sondermann \cite{10}:
%   (1.1)
\begin{equation}\label{4-1.1}
    \text{minimize} \;\;\; E\Bigg( H-x-\int_0^T \tht_t d X_t\Bigg)^2 
        \;\;\;\text{over all} \;\;\; \tht\in \Tht,
\end{equation}
where contingent claim $H$ is a $\cF_T$-measurable square-integrable random variable (r.v.), $x$ is an initial investment, $\Tht$ is a class of admissible trading strategies, $T$ is an investment horizon. 

The mean-variance formulation by Markowitz \cite{27}, provides a foundation for a single period portfolio selection (see, also Merton \cite{28}). In recent paper of Li and Ng \cite{23} the concept of Markowitz's mean-variance formulation for finding the optimal portfolio policy and determining the efficient frontier in analytical form has been extended to multiperiod portfolio selection.

As it pointed out in Li and Ng \cite{23} the results on multiperiod mean-variance formulation with one riskless asset can be derived using the results of the mean-variance hedging formulation.

Therefore, the mean-variance hedging is s powerful approach for both above mentioned major problems.

The problem (\ref{4-1.1}) was intensively investigated in last decade (see, e.g., Dufiie and Richardson \cite{9}, Schwezer \cite{37}, \cite{38}, \cite{39}, Delbaen et al. \cite{8}, Monat and Striker \cite{29}, Rheinl\"{a}nder and Schweizer \cite{34}, (RSch hereafter), Pham et al. \cite{32}, Gourieroux et al. \cite{11} (GLP hereafter), Laurent and Pham \cite{19}).

A stochastic volatility model, proposed by Hull and White \cite{13} and Scott \cite{40}, where the stock price volatility is an random process, is a popular model of incomplete market, where the mean-variance hedging approach can be used (see, e.g., Laurent and Pham \cite{19}, Biagini et al. \cite{13}, Mania and Tevzadze \cite{25}, Pham et al. \cite{32}).

Consider the stochastic volatility model described by the following system of SDE 
%   (1.2)
\begin{equation}\label{4-1.2}
\begin{gathered}
    d X_t=X_t \, dR_t, \quad X_0>0, \\
    d R_t=\mu_t(R_t,Y_t)\,dt +\sg_. dw_t^R, \quad R_0=0, \\
    \sg_t^2=f(Y_t), \\
    dY_t=a(t,Y_t;\al)\,dt+\ve\,d w_t^\sg, \quad Y_0=0,
\end{gathered}
\end{equation}
where $w=(w^R, w^\sg)$ is a standard two-dimensional Wiener process, defined on complete probability space $(\Om,\cF, P)$, $F^w=(\cF_t^w)_{0\leq t\leq T}$ is the $P$-augmentation of the natural filtration $\cF_t^w=\sg(w_s, 0\leq s\leq t)$, $0\leq t\leq T$, generated by $w$, $f(\cdot)$ is a continuous one-to-one positive locally bounded function (e.g., $f(x)=e^x$), $\al=(\al_1,\dots,\al_m)$, $m\geq 1$, is a vector of unknown parameters, and $\ve$, $0<\ve\ll 1$, is a small number. Assume that the system (\ref{4-1.2}) has an unique strong solution.

This model is analogous to the model proposed by Renault and Touzi \cite{33} (RT hereafter). The principal difference is the presence of small parameter $\ve$ in our model, which due to the assumption that the volatility of randomly fluctuated volatility process is small (see, also Sircar and Papanicolau \cite{41}). Thus assumption enables us to use the prices of trading options with short, nearest to the current time value maturities for volatility process filtration and parameter estimation purposes (see below). In contrast, RT \cite{33} needs to assume that there exist trading derivatives with any (up to the infinity) maturities.

Important feature of the stochastic volatility models is that volatility process $Y$ is unobservable (latent) process. To obtain explicit form of optimal trading strategy full knowledge of the model of the process $Y$ is necessary and hence one needs to estimate the unknown parameter $\al=(\al_1,\dots,\al_m)$, $m\geq 1$.

A variety of estimation procedures are used, which involve either direct statistical analysis of the historical data or the use of implied volatilities extracted from prices of existing traded derivatives.

For example, one can use the following method based on historical data.

Fix the time variable $t$. From observations $X_{t_0^{(n)}},\dots,X_{t_n^{(n)}}$, $0=t_0^{(n)}<\cdots<t_n^{(n)}=t$, $\max\limits_j [t_{j+1}^{(n)}-t_j^{(n)}] \to 0$, as $n\to 0$, calculate the realization of yield process $R_t=\int\limits_0^t \frac{dX_s}{X_s}$, and then calculate the sum
$$  
    S_n(t) =\sum_{j=0}^{n-1} |R_{t_{j+1}^{(n)}} -R_{t_j^{(n)}}|^2.
$$

It is well-known (see, e.g., Lipster and Shiryaev \cite{24}) that 
$$  
    S_n(t) \str{P}{\to} \int_0^t \sg_s^2\,ds \;\;\;\text{as} \;\;\; n\to \infty.
$$

Since $\sg_t^2(\om)=f(Y_t)$ is a continuous process we get 
$$  
    \sg_t^2(\om) =\lim_{\Dl \downarrow 0} \frac{F(t+\Dl,\om)-F(t,\om)}{\Dl}\,,
$$
where $F(t,\om)=\int\limits_0^t \sg_s^2(\om)ds$.

Hence, the realization $(y_t)_{0\leq t\leq T}$ of the process $Y$ can be found by the formula $y_t=f^{-1}(\sg_t^2)$, $0\leq t\leq T$.

More sofisticated  methods using the same idea can be found, e.g., in Chesney et al. \cite{5}, Pastorello \cite{31}. 

We can use the reconstructed sample path $(y_t)$, $0\leq t\leq T$. to estimate the unknown parameter $\al$ in the drift coefficient of diffusion process $Y$. 

The second, market price adjusted procedure of reconstruction the sample path of volatility process $Y$ and parameter estimate was suggested by RT \cite{33}, where they used implied volatility data. 

We present a quick review of this method, adapted to our model (\ref{4-1.2}).

Suppose that the volatility risk premium $\lb^\sg\equiv 0$, meaning that the risk from the volatility process is non-compensated (or can be diversified away). Then the price $C_t(\sg)$ of European call option can be calculated by the Hull and White formula (see, e.g., RT \cite{33}), and Black-Scholes (BS) implied volatility $\sg^i(\sg)$ can be found as an unique solution of the equation 
$$  
    C_t(\sg)=C_t^{BS}(\sg^i(\sg)),
$$
where $C^{BS}(\sg)$ denotes the standard BS formula written as a function of the volatility parameter $\sg$. 

Here (for further estimational purposes) only at-the-money options are used.

Under some technical assumptions (see Proposition 5.1 of RT \cite{33}, and Bujeux and Rochet \cite{24} for general diffusion of volatility process)
%   (1.3)
\begin{equation}\label{4-1.3}
    \frac{\pa\sg_t^i(\sg,\al)}{\pa \sg_t}>0
\end{equation}
(remember that the drift coefficient of process $Y$ depends on unknown parameter $\al$).

Fix current value of time parameter $t$, $0\leq t\leq T$, and let $0<T_1<T_2<\cdots<T_{k-1}<t<T_k$ be the maturity times of some traded at-the-money options. 

Let $\sg_{t_j^\ve}^{i^*}$ be the observations of an implied volatility at the time moments $0=t_0^\ve<t_1^\ve<\cdots<t_{[\frac{t}{\ve}]}=t$, $\max\limits_j [t_{j+1}^\ve-t_j^\ve]\to 0$, as $\ve\to 0$.

Then, using (\ref{4-1.3}), and solving the equation 
$$ 
    \sg_{t_j^\ve}^i(\sg_{t_j^\ve},\al)=\sg_{t_j^\ve}^{i^*},
$$
one can obtained the realization $\{\wt\sg_{t_j^\ve}\}$ of the volatility $(\sg_t)$, and thus, using the formula $y_{t_j^\ve}=f^{-1}(\wt\sg_{t_j^\ve}^2)$, the realization $\{y_{t_j^\ve}\}$ of volatility process $(Y_t)$, which can be viewed as the realization of nonlinear AR(1) process: 
$$  
    Y_{t_{j+1}^\ve} -Y_{t_j^\ve} =a(t_j^\ve,Y_{t_j^\ve};\al)
        (t_{j+1}^\ve-t_j^\ve) +\ve(w_{t_{j+1}^\ve}^\sg-w_{t_j^\ve}^\sg).
$$

Using the data $\{y_{t_j^\ve}\}$ one can construct the MLE $\wh\al_t^\ve$ of parameter $\al$, see, e.g., Chitashvili et al. \cite{26}, \cite{27}, Lazrieva and Toronjadze \cite{20}.

Remember the scheme of construction of MLE. Rewrite the previous AR(1) process, using obvious simple notation, in form
$$  
    Y_{j+1}-Y_j =a(t_j,Y_j;\al)\Dl +\ve \Dl w_j^\sg.
$$

Then
\begin{align*}
    \frac{\pa}{\pa y}\, P\{Y_{j+1}\leq y \mid Y_j\} & =
        \frac{1}{\sqrt{2\pi\Dl\ve}}\,\exp 
        \left( -\frac{(y-Y_j-a(t_j,Y_j;\al)\Dl)^2}{2\ve^2\Dl} \right) \\
    & =: \vf_{j+1}(y,Y_j;\al),
\end{align*}
and the $\log$-derivative of the likelihood process $\ell_t=(\ell_t^{(1)},\dots,\ell_t^{(m)})$ is given by the relation 
$$  
    \ell_t^{(i)}=\sum_j \ell_{j+1}^{(i)}, \quad i=\ol{1,m},
$$
where 
\begin{gather*}
    \ell_{j+1}^{(i)}(y;\al)=\frac{\pa}{\pa \al_i}\,
        \ln\vf_{j+1}(y,Y_j;\al) \\
    =\frac{1}{\ve^2\Dl}\,(y-Y_j-a(t_j,Y_j;\al)\Dl)
        \dot{a}^{(i)}(t_j,Y_j;\al)\Dl.
\end{gather*}
Hence MLE is a solution (under some conditions) of the system of equations
$$  
    \frac{1}{\ve^2\Dl} \sum_j (y_{j+1}-y_j-a(t_j,y_j;\al)\Dl)
        \dot{a}^{(i)}(t_j,y_j;\al)\Dl =0, \quad i=\ol{1,m},
$$
where the reconstructed data $\{y_j\}=\{y_{t_j^\ve}\}$ are substituted).

Following RT \cite{33} let us introduce the functionals
\begin{gather*}
    HW_\ve^{-1} :\wh\al_t^\ve(p) \to \left( y_{t_j^\ve}^{(p+1)},\;\;
        0\leq j\leq \left[\frac{t}{\ve}\right] \right), \\
    MLE_\ve :\left( y_{t_j^\ve}^{(p+1)},\;\;
        0\leq j\leq \left[\frac{t}{\ve}\right] \right)\to \wh\al_t^\ve(p+1) 
\end{gather*}
and 
$$  
    \phi_\ve=MLE_\ve \circ HW_\ve^{-1}.
$$

Starting with some constant initial value (or preliminary estimate obtained, e.g., from historical data) one can compute a sequence of estimates 
$$  
    \wh\al_t^\ve(p+1) =\phi_\ve(\wt\al_t^\ve(p)), \quad p\geq 1.
$$
If the operator $\phi_\ve$ is a strong contraction in the neighborhood of the true value of the parameter $\al^0$, for a small enough $\ve$, then one can define the estimate $\wh\al_t^\ve$ as the limits of the sequence $\{\wh \al_t^\ve(p)\}_{p\geq 1}$. It was proved in RT \cite{33} that $\wh \al_t^\ve$ is a strong consistent estimate of the parameter $\al$. 

Return to our consideration.

Interpolating on some way the corresponding (to the estimate $\wh \al_t^\ve$) realization $\{y_{t_j^\ve}\}$ we get the reconstructed continuous sample path $(y_s)_{0\leq s\leq t}$ of the latent process $Y$, which can be used for further analysis.

Unfortunately, both described statistical procedures are highly sensitive w.r.t errors in all steps of parameter identification process.

Hence, this is a natural place for introducing the robust procedure of parameter estimates.

Suppose that the sample path $(y_s)_{0\leq s\leq t}$ comes from the observation of process $(\wt Y_s)_{0\leq s\leq t}$ with distribution $\wt P_\al^\ve$ from the shrinking contamination neighborhood of the distribution $P_\al^\ve$ of the basic process $Y=(Y_s)_{0\leq s\leq t}$. That is 
%   (1.4)
\begin{equation}\label{4-1.4}
    \frac{d \wt P_\al^\ve}{d P_\al^\ve} \,\Big|\, \cF_t^w=\cE_t(\ve N^\ve),
\end{equation}
where $N^\ve=(N_s^\ve)_{0\leq s\leq t}$ is a $P_\al^\ve$-square integrable martingale, $\cE_t(M)$ is the Dolean exponential of martingale $M$.

In the diffusion-type framework (\ref{4-1.4}) represents the Huber gross error model (as it explain in Remark 2.2). The model of type (\ref{4-1.4}) of contamination of measures for statistical models with filtration was suggested by Lazrieva and Toronjadze \cite{21}, \cite{22}.

In Section 2 we study the problem of construction of robust estimates for contamination model (\ref{4-1.4}). 

In subsection 2.1 we give a description of the basic model and definition of consistent uniformly linear asymptotically normal (CULAN) estimates, connected with the basic model (Definition 2.1).

In subsection 2.2 we introduce a notion of shrinking contamination neighborhood, described in terms of contamination of nominal distribution, which naturally leads to the class of alternative measures (see (\ref{4-2.18}) and (\ref{4-2.19})). 

In subsection 2.3 we study the asymptotic behaviour of CULAN estimates under alternative measures (Proposition 2.2), which is the basis for the formulation of the optimization problem.

In subsection 2.4 the optimization problem is solved which leads to construction of optimal $B$-robust estimate (Theorem 2.1).

Based on the limit theorem (subsection 2.1), one can construct the asymptotic confidence region of level $\gm$ for unknown parameter $\al$
$$  
    \lim_{\ve\to \infty} P_\al^\ve \left(\ve^{-2}(\al-\al_t^{*,\ve})' V^{-1}
        (\psi^*;\al_t^{*,\ve})(\al-\al_t^{*,\ve}) \leq \chi_\gm^2\right)=1-\gm,
$$
where $\chi_\gm^2$ is a quantile of order $1-\gm$ of $\chi^2$-distribution with $m$ degree of freedom, and $V(\psi^*;\al)$ is given by (\ref{4-2.17}).

This region shrinks to the estimate $\al_t^{*,0}$, as $\ve\to 0$.

Now if the coefficient $a(t,y;\al)$ in (\ref{4-1.2}) is such that the solution $Y_t^\ve(\al)$ of SDE (\ref{4-1.2}) is continuous w.r.t parameter $\al$ (see, e.g., Krylov \cite{16}), then the confidence region of parameter $\al$ is mapped to the confidence interval for $Y_t^\ve(\al)$, which shrinks to $Y_t^*=Y_t^0(\al_t^{*,0})$, Further, by the function $f$, the latter interval is mapped to the confidence interval for $\sg_t$, which shrinks to $\sg_t^*=f^{1/2}(Y_t^0(\al_t^{*,0}))$. Denote $\sg_t^0$ the center of this interval. Then the interval can be written in the form 
$$  
    \sg_t=\sg_t^0 +\dl(\ve)h_t,
$$
where $\dl(\ve)\to 0$, as $\ve\to 0$, and $h\in \cH$ (see (\ref{4-3.18})). 

Thus, we arrive at the asset price model (\ref{4-1.2}) with misspecified volatility, and it is natural to consider the problem of construction of the robust trading strategy to hedge a contingent claim $H$. 

We investigate this problem in the mean-variance setting in Section 3. We consider the general situation, when the asset price is modelled by $d$-di\-men\-sional continuous semimartingale and the information structure is given by some general filtration. 

In subsection 3.1 we give a description of the financial market model. 

In subsection 3.2 we collect the facts concerning the variance-optimal equivalent local martingale measure, which plays a key role in the mean-variance hedging approach. 

In last subsection 3.3 we construct ``optimal robust hedging strategy'' (Theorem 3.1) by approximating the optimization problem (\ref{4-3.24}) by the problem (\ref{4-3.26}). As it is mentioned in Remark 3.2, such approach and term are common in robust statistic theory. In contact to optimal $B$-robustness (see Section 2), here we develop the approach, known in robust statistics as optimal $V$-robustness, see Hampel et al. \cite{12}.

Note that our approach allows incorporating current information on the underlying model, and hence is adaptive. Namely, passing from time value $t$ to $t+\tau$, $\tau>0$, when more information about market prices are available, the asymptotic variance-covariance of the constructed estimate $\al_t^{*,\ve}$ becomes smaller, and hence the estimation procedure becomes more precise. 

In the paper of Runggaldier and Zaccaria \cite{36} the adaptive approach to risk management under general uncertainty (restricted information) was developed. As it is mentioned in this paper there exist a series of investigations dealt with various type of adaptive approaches (see list of references in \cite{36}). But in all these papers (except Runggaldier and Zaccaria \cite{36}) the uncertainty is only in the stock appreciation rate in contrast to our consideration, where the model misspecification is due to the volatility parameter.

The consideration of misspecified asset price models was initiated by Avella\-neda et al. \cite{1}, Avellaneda and Paras \cite{2}.

Various authors in different settings attacked the robustness problem. The method used in Section 3 was suggested by Toronjadze \cite{42} for asset price process modelled by the one-dimensional process. As it will be shown in Remark 3.2 below, in simplest case when asset price process is a martingale w.r.t initial measure $P$, and it is possible to find the solution of ``exact'' optimization problem (\ref{4-3.24}), this solution coincides with the solution of an approximating optimization problem (\ref{4-3.26}). In more general situation (when asset price process is not more the $P$-martingale) investigation of the problem (\ref{4-3.24}) by, e.g., control theory methods seems sufficiently difficult. Anyway, we do not know the solution of the problem (\ref{4-3.24}).

Return to the stochastic volatility model (\ref{4-1.2}) and describe successive steps of our approach: 

1) For each current time value $t$, $0<t<T$, reconstruct the sample path $(y_s)_{0\leq s\leq t}$, using the historical data or the tradable derivatives prices;

2) Using the approach developed in Section 2, calculate the value $\al_t^{*,\ve}$ of the robust estimate of parameter $\al$ (i.e. construct the deterministic function $t\to \al_t^{*,\ve}\in R^m$) and then find the confidence region for parameter $\al$; 

3) Based on the volatility process model find the confidence interval for $Y_t(\al)$;

4) Denoting $a^*(t,y)=a(t,y;\al_t^{*,\ve})$, where $a(t,y;\al)$ is a drift coefficient of volatility process, consider the stochastic volatility model with misspecified asset price model and fully specified volatility process model
\begin{gather*}
    dX_t=X_t\, dR_t, \quad X_0>0, \\
    dR_t=(\sg_t^0+\dl(\ve)h_t) dM_t^0, \quad R_0=0, \\
    dY_t=a^*(t,Y_t)\,dt+\ve\,d w_t^\sg, \quad Y_0=0, \quad 0\leq t\leq T,
\end{gather*}
where
$$  
    dM_t^0=k_t\,dt +d w_t^R,
$$
$h\in \cH$ and $\sg_t^0$ is the center of the confidence interval of volatility.

Using Theorem 3.1 construct the optimal robust hedging strategy by the formula (\ref{4-3.38}),
$$  
    \tht_t^*=\frac{1}{\sg_t^0}\, \left[ \psi_t^{1,H} +\zt_t(V_t^*-
        (\psi_t^H)'U_t\right],
$$
where all objects are defined in Theorem 3.1.
\hfill \qed

It should be mentioned that if one constructs a hedging strategy $\wt\tht_t^*$ by the above-given formula with $\sg_t^{*,\ve}=f^{1/2}(Y_t^\ve(\al_t^{*,\ve}))$ instead of $\sg_t^0$, then the strategies $\wt\tht_t^*$ and $\tht_t^*$ would be different, since $\sg_t^{*,\ve}\neq \sg_t^0$, in general. Hence the value $\Dl_t=|\sg_t^{*,\ve}-\sg_t^0|$ defines the correction term between the robust, $\tht_t^*$ and non-robust, $\wt\tht^*$ strategies.

In nontrivial case, when $k_t=k(Y_t)$ the variance-optimal martingale measure $\wt P$ is given by (\ref{4-3.17}), $\zt_t=-k_t\cE_t(-k\cdot M^0)$ (see subsection 3.2), and the process $(X_t,Y_t)_{0\leq t\leq T}$ is the Markov process. If $H=h(X_T,Y_T)$ $(h(x,y)$ is some function), then $\wt V_t^H=E^{\wt P}(H|\cF_t^w)=E^{\wt P}(h(X_T,Y_T)|\cF_t^w)=v(t,X_T,Y_T)$ and if, e.g., $v(t,x,y)\in C^{1,2,2}$, then $v$ is an unique solution of the following partial differential equation 
$$  
    \frac{\pa v}{\pa t} +a^*\,\frac{\pa v}{\pa y}+\frac{1}{2} 
        \left(\ve^2\,\frac{\pa^2v}{\pa y^2}+x^2v^2\,\frac{\pa^2 v}{\pa x^2}\right)=0,
$$
with the boundary condition $v(T,x,t)=h(x,y)$. More general situation with nonsmooth $v$ is considered in Laurent and Pham \cite{19}, Mania and Tevzadze \cite{25}.

Further, one can find the Galtchouck--Kunita--Watanabe decomposition of r.v. $H$ (see, e.g., Pham  et al. \cite{32}) putting
$$  
    \xi_t^H=\frac{\pa v(t,X_t,Y_t)}{\pa x}\,, \quad 
    L_T^H=\ve \int_0^T \frac{\pa v}{\pa y}\,(t,X_t,Y_t)\,d w_t^\sg,
$$
and calculate $\psi_t^H$, $L_T$ and $V_t^*$ using (4.13) and (4.14) of RSch \cite{34}. 

Thus one get the explicit solution of the mean-variance hedging problem.

Finally, here is the short summary of approach:

a) Incorporate the robust procedure in statistical analysis of volatility process. That is construct and use in the model optimal $B$-robust estimate of unknown parameter in drift  coefficient of volatility process.

Parameter estimation naturally leads to the asset price model misspecification. 

b) Incorporate the second robust procedure in financial analysis of contingent claim hedging. That is construct and use for hedging purposes optimal $V$-robust trading strategy.

In our opinion this ``double robust'' strategy should be more attractive to protect the hedger against the possible errors.

The general asymptotic theory of estimation can be found in Ibragimov and Khas'miskii \cite{14}; the theory of robust statistics is developed in Hampel et al. \cite{12} and in Rieder \cite{35}; the theory of the trend parameter estimates for diffusion process with small noise is developed in Kutoyants \cite{18}; the book of Musiela and Rutkowsky \cite{30} is devoted to the mathematical theory of finance and finally, the general theory of martingales can be found in Jacod and Shiryaev \cite{15}.

%   2
\section{Optimal $B$-Robust Estimates}

\subsection{Basic model. CULAN estimates}
The basic model of observations is described by the SDE
%   (2.1)
\begin{equation}\label{4-2.1}
    d Y_s=a(s,Y;\al)\,ds +\ve\,dw_s, \quad Y_0=0, \quad 0\leq s\leq t,
\end{equation}
where $t$ is a fixed  number, $w=(w_s)_{0\leq s\leq t}$ is a standard Wiener process defined on the filtered probability space $(\Om,\cF,F=(\cF_s)_{0\leq s\leq t},P)$ satisfying the usual conditions, $\al=(\al_1,\dots,\al_m)$, $m\geq 1$, is an unknown parameter to be estimated, $\al\in \cA\sbs R^m$, $\cA$ is an open subset of $R^m$, $\ve$, $0<\ve\ll 1$, is a small parameter (index of series). In our further considerations all limits correspond to $\ve\to 0$.

Denote $(C_t,\cB_t)$ a measurable space of continuous on $[0,t]$ functions \linebreak $x=(x_s)_{0\leq s\leq t}$ with $\sg$-algebra $\cB_t=\sg(x:x_s,s\leq t)$. Put $\cB_s=\sg(x:x_u,u\leq s)$.

Assume that for each $\al\in \cA$ the drift coefficients $a(s,x;\al)$, $0\leq s\leq t$, $x\in C_t$ is a known nonanticipative (i.e. $\cB_s$-measurable for each $s$, $0\leq s\leq t$) functional satisfying the functional Lipshitz and linear growth conditions $\bL$:
\begin{gather*}
    |a(s,x^1;\al)-a(s,x^2;\al)| \leq L_1 \int_0^s |x_u^1-x_u^2|\,d k_u +
        L_2|x_s^1-x_s^2|, \\
    |a(s,x;\al)|\leq L_1 \int_0^s (1+|x_u|)\,d k_u+L_2(1+|x_s|),
\end{gather*}
where $L_1$ and $L_2$ are constants, which do not depend on $\al$, $k=(k(s))_{0\leq s\leq t}$ is a non-decreasing right-continuous function, $0\leq k(s)\leq k_0$, $0:k_0<\infty$, $x^1,x^2\in C_t$. 

Then, as it is well-known (see, e.g., Liptser and Shiryaev \cite{24}), for each $\al\in \cA$ the equation (\ref{4-2.1}) has an unique strong solution $Y^\ve(\al)=(Y_s^\ve(\al))_{0\leq s\leq t}$, and in addition (see Kutoyants \cite{18})
$$  
    \sup_{0\leq s\leq t} |Y_s^\ve(\al)-Y_s^0(\al)| \leq 
        C\ve \sup_{0\leq s\leq t} |w_s| \;\;\Pas,
$$
with some constant $C=C(L_1,L_2,k_0,t)$, where $Y^0(\al)=(Y_s^0(\al))_{0\leq s\leq t}$ is the solution of the following nonperturbated differential equation 
%   (2.2)
\begin{equation}\label{4-2.2}
    d Y_s=a(s,Y;\al)\,ds, \quad Y_0=0.
\end{equation}

Change the initial problem of estimation of parameter $\al$ by the equivalent one, when the observations are modelled according to the following SDE
%   (2.3)
\begin{equation}\label{4-2.3}
    d X_s=a_\ve(s,X;\al)\,ds +dw_s, \quad X_0=0,
\end{equation}
where $a_\ve(s,x;\al)=\frac{1}{\ve}\,a(s,\ve x;\al)$, $0\leq s\leq t$, $x\in C_t$, $\al\in \cA$.

It is clear that if $X^\ve(\al)=(X_s^\ve(\al))_{0\leq s\leq t}$ is the solution of SDE (\ref{4-2.3}), then for each $s\in [0,t]$ $\ve X_s^\ve(\al)=Y_s^\ve(\al)$.

Denote by $P_\al^\ve$ the distribution of process $X^\ve(\al)$ on the space $(C_t,\cB_t)$, i.e. $P_\al^\ve$ is the probability measure on $(C_t,\cB_t)$ induced by the process $X^\ve(\al)$. Let $P^w$ be a Wiener measure on $(C_t,\cB_t)$. Denote $X=(X_s)_{0\leq s\leq t}$ a coordinate process on $(C_t,\cB_t)$, that is $X_s(x)=x_s$, $x\in C_t$.

The conditions $\bL$ guarantee that for each $\al\in \cA$ the measures $P_\al^\ve$ and $P^w$ are equivalent $(P_\al^\ve\sim P^w)$, and if we denote $z_s^{\al,\ve}=\frac{dP_\al^\ve}{dP^w}|\cB_s$ the density process (likelihood ratio process), then 
$$ 
    z_s^{\al,\ve}(X)=\cE_s(a_\ve(\al) \cdot X) :=
        \exp \Bigg\{ \int_0^s a_\ve(u,X;\al)\, dX_u-
        \frac{1}{2} \int_0^s a_\ve^2(u,X;\al)\,du\Bigg\}.
$$

Introduce class $\Psi$ of $R^m$-valued nonanticipative functionals $\psi$, $\psi:[0,t]\times C_t\times \cA \to R^m$ such that for each $\al\in \cA$ and $\ve>0$
%   (2.4)-(2.6)
\begin{align}
    1) \quad & E_\al^\ve \int_0^t |\psi(s,X;\al)|^2ds<\infty, \label{4-2.4} \\
    2) \quad & \int_0^t |\psi(s,Y^0(\al);\al)|^2ds<\infty, \label{4-2.5} \\
    3) \quad & \text{uniformly in $\al$ on each compact $K\sbs \cA$} \notag \\
    & P_\al^\ve -\lim_{\ve\to 0} \int_0^t |\psi(s,\ve X;\al)-
        \psi(s,Y^0(\al);\al)|^2 ds=0, \label{4-2.6}
\end{align}
where $|\cdot|$ is an Euclidean norm in $R^m$, $P_\al^\ve-\lim\limits_{\ve\to 0} \zt_\ve=\zt$ denotes the convergence $P_\al^\ve\{|\zt_\ve-\zt|>\rho\}\to 0$, as $\ve\to 0$, for all $\rho$, $\rho>0$.

Assume that for each $s\in [0,t]$ and $x\in C_t$ the functional $a(s,x;a)$ is differentiable in $\al$ and gradient $\dot{a}=\big(\frac{\pa}{\pa \al_1}\,a,\dots,\frac{\pa}{\pa\al_m}\,a\big)'$ belongs to $\Psi$ $(\dot{a}\in \Psi)$, where the sign ``$'$'' denoted a transposition. 

Then the Fisher information process 
$$  
    I_s^\ve(X;\al) :=\int_0^s \dot{a}_\ve(u,X;\al) 
        [\dot{a}_\ve(u,X;\al)]'du, \quad 0\leq s\leq t,
$$
is well-defined and, moreover, uniformly in $\al$ on each compact 
%   (2.7)
\begin{equation}\label{4-2.7}
    P_\al^\ve-\lim_{\ve\to 0} \ve^2 I_t^\ve(\al)=I^0(\al),
\end{equation}
where
$$  
    I^0(\al):= \int_0^t \dot{a}(s,Y^0(\al);\al) 
        [\dot{a}(s,Y^0(\al);\al]'ds.
$$

For each $\psi\in \Psi$, introduce the functional $\psi_\ve(s,x;\al):=\frac{1}{\ve} \,\psi(s,\ve x;\al)$ and matrices $\Gm_\ve^\psi(\al)$ and $\gm_\ve^\psi\al$:
\allowdisplaybreaks
%   (2.8)-(2.9)
\begin{align}
    \Gm_\ve^\psi(X;\al) & := \int_0^t \psi_\ve(s,X;\al) [\psi_\ve(s,X;\al)]'ds,
        \label{4-2.8} \\
    \gm_\ve^\psi(X;\al) & := \int_0^t \psi_\ve(s,X;\al) [\dot{a}_\ve(s,X;\al)]'ds.
        \label{4-2.9} 
\end{align}

Then from (\ref{4-2.6}) it follows that uniformly in $\al$ on each compact 
%   (2.10)-(2.11)
\begin{align}
    & P_\al^\ve -\lim_{\ve\to 0} \ve^2 \Gm_\ve^\psi(\al) =\Gm_0^\psi(\al),
        \label{4-2.10} \\
    & P_\al^\ve -\lim_{\ve\to 0} \ve^2 \gm_\ve^\psi(\al) =\gm_0^\psi(\al),
        \label{4-2.11} 
\end{align}
where the matrices $\Gm_0^\psi(\al)$ and $\gm_0^\psi(\al)$ are defined as follows
%   (2.12)-(2.13)
\begin{align}
    \Gm_0^\psi(\al) & := \int_0^t \psi(s,Y^0(\al);\al) [\psi(s,Y^0(\al);\al)]'ds,
        \label{4-2.12} \\
    \gm_0^\psi(\al) & := \int_0^t \psi(s,Y^0(\al);\al) 
        [\dot{a}(s,Y^0(\al);\al)]'ds.
        \label{4-2.13} 
\end{align}

Note that,by virtue of (\ref{4-2.4}), (\ref{4-2.5}) and $\dot{a}\in \Psi$, matrices given by (\ref{4-2.8}), (\ref{4-2.9}), (\ref{4-2.12}) and (\ref{4-2.13}) are well-defined.

Denote $\Psi_0$ the subset of $\Psi$ such that for each $\psi\!\in \!\Psi_0$ and $\al\in \cA$, $\rank \Gm_0^\psi(\al)\!=m$ and $\rank \gm_0^\psi(\al)=m$. 

Assume that $\dot{a}\in \Psi_0$.

For each $\psi\in \Psi_0$, define a $P_\al^\ve$-square integrable martingale $L^{\psi,\ve}(\al)= \linebreak (L_s^{\psi,\ve}(\al))_{0\leq s\leq t}$ as follows
%   (2.14)
\begin{equation}\label{4-2.14}
    L_s^{\psi,\ve}(X;\al) =\int_0^s \psi_\ve(u,X;\al) (d X_u-\al_\ve(u,X;\al)\,du).
\end{equation}

Now we give a definition of CULAN $M$-estimates.

%   Definition 2.1
\begin{definition}\label{d4-2.1}
An estimate $(\al_t^{\psi,\ve})_{\ve>0} =(\al_{1,t}^{\psi,\ve}, \dots, \al_{m,t}^{\psi,\ve})_{\ve>0}'$, $\psi\in \Psi_0$, is called consistent uniformly lineal asymptotically normal (CULAN) if it admits the following expansion
%   (2.15)
\begin{equation}\label{4-2.15}
    \al_t^{\psi,\ve}=\al+[\gm_0^\psi(\al)]^{-1} \ve^2 L_t^{\psi,\ve}(\al) +
        r_{\psi,\ve}(\al),
\end{equation}
where uniformly in $\al$ on each compact
%   (2.16)
\begin{equation}\label{4-2.16}
    P_\al^\ve-\lim_{\ve\to 0} \ve^{-1} r_{\psi,\ve}(\al)=0.
\end{equation}
\end{definition}

It is well-known (see Lazrieva, Toronjadze \cite{20}) that under the above conditions uniformly in $\al$ on each compact 
$$  
    \cL\{\ve^{-1}(\al_t^{\psi,\ve}-\al) \mid P_\al^\ve\} \str{w}{\to} 
        N(0,V(\psi;\al)),
$$
with
%   (2.17)
\begin{equation}\label{4-2.17}
    V(\psi;\al):=[\gm_0^\psi(\al)]^{-1} \Gm_0^\psi(\al) 
        ([\gm_0^\psi(\al)]^{-1})',
\end{equation}
where $\cL(\zt|P)$ denotes the distribution of random vector $\zt$ calculated under measure $P$, symbol ``$\str{w}{\to}$'' denotes the weak convergence of measures, \linebreak  $N(0,V(\psi;\al))$ is a distribution of Gaussian vector with zero mean and covariance matrix $V(\psi;\al)$.

%   Remark 2.1
\begin{remark}\label{r4-2.1}
In context of diffusion type processes the $M$-estimate $(\al_t^{\psi,\ve})_{\ve>0}$ is defined as a solution of the following stochastic equation 
$$  
    L_t^{\psi,\ve}(X;\al)=0,
$$
where $L_t^{\psi,\ve}(X;\al)$ is defined by (\ref{4-2.14}), $\psi\in \Psi_0$.
\end{remark}

The asymptotic theory of $M$-estimates for general statistical models with filtration is developed in Chitashvili et al. \cite{7}. Namely, the problem of existence and global behaviour of solutions is studied. In particular, the conditions of regularity and ergodicity type are established, under which $M$-estimates have a CULAN property. 

For our model, in case when $\cA=R^m$, the sufficient conditions for CULAN property take the form: 

(1) for all $s$, $0\leq s\leq t$, and $x\in C_t$ the functionals $\psi(s,x;\al)$ and $\dot{a}(s,x;\al)$ are twice continuously differentiable in $\al$ with bounded derivatives satisfying the functional Lipshitz conditions with constants, which do not depend on $\al$. 

(2) the equation (w.r.t $y$)
$$  
    \Dl(\al,y):= \int_0^t \psi(s,Y^0(\al);y) 
        (a(s,Y^0(\al);\al)-a(s,Y^0(\al);y))\,ds=0,
$$
has an unique solution $y=\al$.

The MLE is a special case of $M$-estimates when $\psi=\dot{a}$.

%   Remark 2.2
\begin{remark}\label{r4-2.2}
According to (\ref{4-2.7}) the asymptotic covariance matrix of MLE $(\wh\al_t^\ve)_{\ve>0}$ is $[I_0(\al)]^{-1}$. By the usual technique one can show that for each $\al\in \cA$ and $\psi\in \Psi_0$, $I_0^{-1}(\al)\leq V(\psi,\al)$ (see (\ref{4-2.17})), where for two symmetric matrices $B$ and $C$ the relation $B\leq C$ means that the mattix $C-B$ is nonnegative definite.
\end{remark}

Thus, the MLE has a minimal covariance matrix among all $M$-estimates.

%   2.2
\subsection{Shrinking contamination neighborhoods}
In this subsection we give a notion of a contamination of the basic model (\ref{4-2.3}), described in terms of shrinking neighborhoods of basic measures $\{P_\al^\ve$, $\al\in \cA$, $\ve>0\}$, which is an analog of the Huber gross error model (see, e.g., Hampel et.al. \cite{12} and also, Remark \ref{r4-2.3} below). 

Let $\cH$ be a family of bounded nonanticipative functionals $h:[0,t]\times C_t\times \cA\to R^1$ such that for all $s\in [0,t]$ and $\al\in \cA$ the functional $h(s,x;\al)$ is continuous at the point $x_0=Y^0(\al)$.

Let for each $h\in \cH$, $\al\in\cA$ and $\ve>0$, $P_\al^{\ve,h}$ be a measure on $(C_t,\cB_t)$ such that 
%   (2.18)
\begin{align}
    1) \qquad & P_\al^{\ve,h}\sim P_\al^\ve, \notag  \\ 
    2) \qquad & \frac{dP_\al^{\ve,h}}{d P_\al^\ve} =\cE_t(\ve N_\al^{\ve,h}),
        \label{4-2.18}
\end{align}
where 
%   (2.19)
\begin{equation}\label{4-2.19}
    3) \qquad N_{\al,s}^{\ve,h}:= \int_0^s h_s(u,X;\al) 
        (dX_u-a_\ve(u,X;\al)\,du),
\end{equation}
with $h_\ve(s,x;\al):=\frac{1}{\ve}\,h(s,\ve x;\al)$, $0\leq s\leq t$, $x\in C_t$.

Denote $\bP_\al^{\ve,\cH}$ a class of measures $P_\al^{\ve,h}$, $h\in \cH$, that is 
$$  
    \bP_\al^{\ve,\cH} =\{P_\al^{\ve,h};\;\;h\in \cH\}.
$$
We call $(\bP_\al^{\ve,\cH})_{\ve>0}$ a shrinking contamination neighborhoods of the basic measures $(P_\al^\ve)_{\ve>0}$, and the element $(P_\al^{\ve,h})_{\ve>0}$ of these neighborhoods is called alternative measure (or simply alternative).

Obviously for each $h\in \cH$ and $\al\in \cA$, the process $N_\al^{\ve,h}=(N_{\al,s}^{\ve,h})_{0\leq s\leq t}$ defined by (\ref{4-2.19}) is a $P_\al^\ve$-square integrable martingale. Since under measure $P_\al^\ve$ the process $\ol{w}=(\ol{w}_s)_{0\leq s\leq t}$ defined as 
$$  
    \ol{w}_s:= X_s-\int_0^s a_\ve(u,X;\al)\,du, \quad 0\leq s\leq t,
$$
is a Wiener process, then by virtue of the Girsanov Theorem the process $\wt w:=\ol{w}+\la \ol{w},\ve N_\al^{\ve,h}\ra$ is a Wiener process under changed measure $P_\al^{\ve,h}$. But by the definition 
$$  
    \wt w_s=X_s-\int_0^s (a_\ve(u,X;\al)+\ve h_\ve(u,X;\al))\,du,
$$
and hence, one can conclude that $P_\al^{\ve,h}$ is a weak solution of SDE
$$  
    d X_s=(a_\ve(s,X;\al)+\ve h_\ve(s,X;\al))\,ds +d w_s, \quad X_0=0.
$$ 

This SDE can be viewed as a ``small'' perturbation of the basic model (\ref{4-2.3}).

%   Remark 2.3
\begin{remark}\label{r4-2.3}
1) In the case of i.i.d. observations $X_1,X_2,\dots,X_n$, $n\geq 1$, the Huber gross error model in shrinking setting is defined as follows
$$  
    f^{n,h}(x;\al):=(1-\ve_n)f(x;\al)+\ve_n h(x;\al),
$$
where $f(x;\al)$ is a basic (core) density of distribution of r.v. $X_i$ (w.r.t some dominating measure $\mu$), $h(x;\al)$ is a contaminating density, $f^{n,h}(x;\al)$ is a contaminated density, $\ve_n=O(n^{-1/2})$. If we denote by $P_\al^n$ and $P_\al^{n,h}$ the measures on $(R^n,\cB(R^n))$, generated by $f(x;\al)$ and $f^{n,h}(x;\al)$, respectively, then 
$$  
    \frac{dP_\al^{n,h}}{d P_\al^n} =\prod_{i=1}^n 
        \frac{f^{n,h}(X_i;\al)}{f(X_i;\al)} =
        \prod_{i=1}^n (1+\ve_n H(X_i;\al))=\cE_n(\ve_n\cdot N_\al^{n,h}),
$$
where $H=\frac{h-f}{f}$, $N_\al^{n,h}=(N_{\al,m}^{n,h})_{1\leq m\leq n}$, $N_{\al,m}^{n,h}=\sum\limits_{i=1}^m H(X_i;\al)$, $N_\al^{n,h}$ is a $P_\al^n$-martingale, $\cE_n(\ve_n N_\al^{n,h})=\prod\limits_{i=1}^n (1+\ve_n\Dl N_{\al,i}^{n,h})$ is the Dolean exponential in discrete time case.

Thus 
%   (2.20)
\begin{equation}\label{4-2.20}
    \frac{dP_\al^{n,h}}{d P_\al^n}=\cE(\ve_n\cdot N_\al^{n,n}),
\end{equation}
and the relation (\ref{4-2.18}) is a direct analog of (\ref{4-2.20}).

2) The concept of shrinking contamination neighborhoods, expressed in the form of (\ref{4-2.18}) was proposed in Lazrieva and Toronjadze \cite{21} for more general situation, concerning with the contamination areas for semimartingale statistical models with filtration.
\hfill \qed 
\end{remark}

Note here that the power of the small parameter $\ve$ is crucial. One cannot consider the perturbation of measure with different power of $\ve$ if he/she wish to get nontrivial result.

In the remainder of this subsection we study the asymptotic properties of CULAN estimates under alternatives.

For this aim we first consider the problem of contiguity of measures $(P_\al^{\ve,h})_{\ve>0}$ to $(P_\al^\ve)_{\ve>0}$.

Let $(\ve_n)_{n\geq 1}$, $\ve_n \downarrow 0$, and $(\al_n)_{n\geq 1}$, $\al_n\in K$, $K\sbs \cA$ is a compact, be arbitrary sequences.

%   proposition 2.1
\begin{proposition}\label{p4-2.1}
For each $h\in \cH$ the sequence of measures $(P_{\al_n}^{\ve_n,h})$ is contiguous to sequence of measures $(P_{\al_n}^{\ve_n})$, i.e.
$$  (P_{\al_n}^{\ve_n,h}) \vartriangleleft (P_{\al_n}^{\ve_n}). $$
\end{proposition}

\begin{proof}
From the predictable criteria of contiguity (see, e.g., Jacod and Shiryaev \cite{15}), follows that we have to verify the relation 
%   (2.21)
\begin{equation}\label{4-2.21}
    \lim_{N\to \infty} \limsup_{n\to \infty} 
        P_{\al_n}^{\ve_n,h} \left\{ h_t^n \left(\frac{1}{2}\right)>N\right\}=0,
\end{equation}
where $h^n(\frac{1}{2})=(h_s^n(\frac{1}{2}))_{0\leq s\leq t}$ is the Hellinger process of order $\frac{1}{2}$.

By the definition of Hellinger process (see, e.g., Jacod and Shiryaev \cite{15}) we have
$$  
    h_t^n\left(\frac{1}{2}\right) =h_t^n \left(\frac{1}{2}\,, 
        P_{\al_n}^{\ve_n,h}, P_{\al_n}^{\ve_n}\right) =
        \frac{1}{8} \int_0^t \left[ h(s,\ve_n X;\al_n)\right]^2ds,
$$
and since $h\in \cH$, and hence is bounded, $h_t^n(\frac{1}{2})$ is bounded too, which provides (\ref{4-2.21}).
\end{proof}

%   Proposition 2.2
\begin{proposition}\label{p4-2.2}
For each estimate $(\al_t^{\ve,\psi})_{\ve>0}$ with $\psi\in \Psi_0$ and each alternative $(P_\al^{\ve,h})_{\ve>0} \in (\bP_\al^{\ve,\cH})_{\ve>0}$ the following relation holds true
$$  
    \cL\left\{ \ve^{-1} (\al_t^{\psi,\ve}-\al) \mid P_\al^{\ve,h}\right\} 
        \str{w}{\to} N\left( [\gm_0^\psi(\al)]^{-1} b(\psi,h;\al), V(\psi,\al)\right),
$$
where
$$  
    b(\psi,h;\al):=\int_0^t \psi(s,Y^0(\al);\al) h(s,Y^0(\al);\al)\,ds.
$$
\end{proposition}

\begin{proof}
Proposition \ref{p4-2.1} together with (\ref{4-2.16}) provides that uniformly in $\al$ on each compact 
$$  
    P_\al^{\ve,h} -\lim_{\ve\to 0} \ve^{-1} r_{\psi,\ve}(\al)=0,
$$ 
and therefore we have to establish the limit distribution of random vector $[\gm_0^\psi(\al)]^{-1}\ve L_t^{\psi,\ve}$ under the measures $(P_\al^{\ve,h})_{\ve>0}$.

By virtue of the Girsanov Theorem the process $L^{\psi,\ve}(\al)=(L_s^{\psi,\ve}(\al))_{0\leq s\leq t}$ is a semimartingale with canonical decomposition
%   (2.22)
\begin{equation}\label{4-2.22}
    L_s^{\psi,\ve}(\al) =\wt L_s^{\psi,\ve}(\al)+b_{\ve,s}(\psi,h;\al), 
        \quad 0\leq s\leq t,
\end{equation}
where $\wt L^{\psi,\ve}(\al)=(\wt L_s^{\psi,\ve}(\al))_{0\leq s\leq t}$ is a $P_\al^{\ve,h}$-square integrable martingales defined as follows
$$  
    \wt L_s^{\psi,\ve}(X;\al):=\int_0^s \psi_\ve(u,X;\al)\,
        \left(d X_u-(a_\ve(u,X;\al)+\ve h_\ve(u,X;\al)\right)\,du,
$$
and 
$$   
    b_{\ve,s}(\psi,h;\al):= \ve \int_0^s \psi_\ve (u,X;\al) h_\ve(u,X;\al)\,du.
$$ 

But $\la \wt L^{\psi,\ve}(\al)\ra_t=\Gm_\ve^\psi(\al)$, where $\Gm_\ve^\psi(\al)$ is defined by (\ref{4-2.8}). On the other hand, from Proposition \ref{p4-2.1} and (\ref{4-2.10}) it follows that 
$$ 
    P_\al^{\ve,h}-\lim_{\ve\to 0} \la \ve \wt L^{\psi,\ve}(\al)\ra_t =
        P_\al^{\ve,h}-\lim_{\ve\to 0} \ve^2 \Gm_\ve^\psi(\al) =
        P_\al^\ve -\lim_{\ve\to 0} \ve^2 \Gm_\ve^\psi(\al)=\Gm_0^\psi(\al)
$$
uniformly in $\al$ on each compact, and hence 
%   (2.23)
\begin{equation}\label{4-2.23}
    \cL \left\{ [\gm_0^\psi(\al)]^{-1} \ve \wt L_t^{\psi,\ve}(\al) \mid 
        P_\al^{\ve,h}\right\} \str{w}{\to} N(0,V(\psi;\al)).
\end{equation}

Finally, relation (\ref{4-2.23}) together with (\ref{4-2.22}) and relation 
$$  
    P_\tht^{\ve,h} -\lim_{\ve\to 0} \ve b_{\ve,t}(\psi,h;\al) =
        \int_0^t \psi(s,Y^0(\al);\al) h(s,Y^0(\al);\al)\,ds =
        b(\psi,h;\al),
$$
provides the desirable results.
\end{proof}

%   2.3
\subsection{Optimization criteria. Construction of optimal $B$-robust estimates}
In this subsection we state and solve an optimization problem, which results in construction of optimal $B$-robust estimate.

Initially, it should be stressed that the bias vector $\wt b(\psi,h;\al):=[\gm_0^\psi(\al)]^{-1} \times \linebreak b(\psi,h;\al)$ can be viewed as the influence functional of the estimate $(\al_t^{\psi,\ve})_{\ve>0}$ w.r.t. alternative $(P_\al^{\psi,h})_{\ve>0}$.

Indeed, the expansion (\ref{4-2.15}) together with (\ref{4-2.22}) and (\ref{4-2.23}) allows to conclude that 
$$  
    \cL\left\{ \ve^{-1} (\al_t^{\psi,\ve} -\al-\ve^2[\gm_0^\psi(\al)]^{-1} 
        b_\ve(\psi,h;\al)) \mid P_\al^{\ve,h} \right\} \str{w}{\to} 
        N(0,V(\psi,\al)),
$$
and, hence, the expression 
$$  
    \al+\ve^2[\gm_0^\psi(\al)]^{-1} b_\ve(\psi,h;\al)-\al=
        \ve^2 [\gm_0^\psi(\al)]^{-1} b_\ve(\psi,h;\al),
$$
plays the role of bias on the ``fixed step $\ve$'' and it seems natural to interpret the limit
$$  
    P_\al^{\ve,h} -\lim_{\ve\to 0} \frac{\al+\ve^2[\gm_0^\psi(\al)]^{-1} 
        b_\ve(\psi,h;\al)-\al}{\ve} =[\gm_0^\psi(\al)]^{-1} b(\psi,h;\al),
$$
as the influence functional.

For each estimate $(\al_t^{\psi,\ve})_{\ve>0}$, $\psi\in \Psi_0$, define the risk functional w.r.t. alternative $(P_\al^{\ve,h})_{\ve>0}$, $h\in \cH$, as follows:
$$  
    D(\psi,h;\al) =\lim_{a\to \infty} \lim_{\ve\to 0} 
        E_\al^{\ve,h}\left( (\ve^{-2} |\al_t^{\psi\ve}-\al|^2) \wedge a\right),
$$
where $x\wedge \al=\min(x,a)$, $a>0$, $E_\al^{\ve,h}$ is an expectation w.r.t. measure $P_\al^{\ve,h}$. 

Using Proposition \ref{p4-2.2} it is not hard to verify that 
$$  
    D(\psi,h;\al)=|\wt b(\psi,h;\al)|^2+\tr V(\psi,\al),
$$
where $\tr A$ denotes the trace of matrix $A$.

By Proposition \ref{p4-2.2}
$$
    \ve^{-1} (\al_t^{\psi,\ve}-\al)\stackrel{d}{\to}
        N\big(\wt b(\psi,h;\al),V(\psi;\al)\big),
$$
where $\stackrel{d}{\to}$ denotes the convergence by distribution (by distribution $P_\al^{\ve,h}$ in our case),  $N(\wt b,V)$ is a Gaussian random vector with mean $\wt b$ and covariation matrix~$V$.

But if $\xi=(\xi_1,\dots,\xi_m)'$ is a Gaussian vector with parameters $(\mu,\sg^2)$, then
$$
    E|\xi|^2=\sum_{i=1}^m E\xi_i^2=\sum_{i=1}^m (E\xi_i)^2+
        \sum_{i=1}^m D\xi_i=|\mu|^2+\tr \sg^2,
$$
as it was required.

Connect with each $\psi\in \Psi_0$ the function $\wt\psi$ as follows
$$  
    \wt\psi(s,x;\al)=[\gm_0^\psi(\al)]^{-1} \psi(s,x;\al), \quad 
        0\leq s\leq t, \quad x\in C_t, \quad \al\in \cA.
$$
Then $\wt\psi\in \Psi_0$ and 
$$  
    \gm_0^{\wt\psi}(\al)=Id,
$$
where $Id$ is an unit matrix,
$$  
    V(\psi;\al)=V(\wt\psi;\al)=\Gm_0^{\wt\psi}(\al), \quad 
    \wt b(\psi,h;\al)=\wt b(\wt\psi,h;\al)=b(\wt\psi,h;\al).
$$
Therefore 
%   (2.24)
\begin{equation}\label{4-2.24}
    D(\psi,h;\al) =D(\wt\psi,h;\al) =|b(\wt\psi,h;\al)|^2 +
        \tr \Gm_0^{\wt\psi}(\al).
\end{equation}

Denote $\cH_r$, a set of functions $h\in \cH$ such that for each $\al\in \cA$
$$  
    \int_0^t |h(s,Y^0(\al);\al)|\,ds\leq r,
$$
where $r$, $r>0$, is a constant.

Since, for each $r>0$,
$$  
    \sup_{h\in \cH_r} |b(\wt\psi,h;\al)|\leq const(r) 
        \sup_{0\leq s\leq t} |\wt\psi(s,Y^0(\al);\al)|,
$$
where constant depends on $r$, we call the function $\wt \psi$ an influence function of estimate $(\al_t^{\psi,\ve})_{\ve>0}$ and a quantity 
$$  
    \gm_\psi^*(\al) =\sup_{0\leq s\leq t} |\wt\psi(s,Y^0(\al);\al)|
$$
is named as the (unstandardized) gross error sensitivity at point $\al$ of estimate $(\al_t^{\psi,\ve})_{\ve>0}$.

Define 
%   (2.25)-(2.26)
\begin{gather}
    \Psi_{0,c} =\bigg\{\psi\in \Psi_0\;:\;
        \int_0^t \psi(s,Y^0(\al);\al) [\dot{a} (s,Y^0(\al);\al)]'ds=
        Id, \label{4-2.25} \\
    \gm_\psi^*(\al)\leq c\bigg\},    \label{4-2.26}
\end{gather}
where $c\in [0,\infty)$ is a generic constant.

Take into account the expression (\ref{4-2.24}) for the risk functional we come to the following optimization problem, known in robust estimation theory as Hampel's optimization problem: minimize the trace of the asymptotic covariance matrix of estimate $(\al_t^{\psi,\ve})_{\ve>0}$ over the class $\Psi_{0,c}$, that is
%   (2.27)
\begin{equation}\label{4-2.27}
    \text{minimize} \;\;\; \int_0^t \psi(s,Y^0(\al);\al) 
        [\psi(s,Y^0(\al);\al)]'ds
\end{equation}
under the side conditions (\ref{4-2.25}) and (\ref{4-2.26}).

Define the Huber function $h_c(z)$, $z\in R^m$, $c>0$, as follows
$$  
    h_c(z):= z\min\left(1,\frac{c}{|z|}\right).
$$

For arbitrary nondegenerate matrix $A$ denote $\psi_c^A=h_c(A\dot{a})$.

%   Theorem 2.1
\begin{theorem}\label{t4-2.1}
Assume that for given constant $c$ there exists a nondegenerate $m\times m$-matrix $A_c^*(\al)$, which solves the equation $($w.r.t. matrix $A)$
%   (2.28)
\begin{equation}\label{4-2.28}
    \int_0^t \psi_c^A(s,Y^0(\al);\al) [\dot{a}(s,Y^0(\al);\al)]'ds=Id.
\end{equation}

Then the function $\psi_c^{A_c^*(\al)}=h_c(A_c^*(\al)\dot{a})$ solves the optimization problem $(\ref{4-2.27})$.
\end{theorem}

\begin{proof}
We follow Hampel et al. \cite{12}.

Let $A$ be an arbitrary $m\times m$-matrix.

Since for each $\psi\in \Psi_{0,c}$, $\int \psi(\dot{a})'=Id$, $\int \dot{a}[\dot{a}]' =I^0(\al)$ (see (\ref{4-2.7})), then 
$$  
    \int (\psi-A\dot{a} )(\psi-A\dot{a})' =
        \int \psi\psi'-A-A'+AI^0(\al)A'
$$
(here and below we use simple evident notation for integrals). 

Therefore since the trace is an additive functional instead of minimizing of $\tr \int \psi \psi'$ we can minimize 
$$  
    \tr \int(\psi-A\dot{a})(\psi-A\dot{a})'=\int |\psi-A\dot{a}|^2.
$$

Note that for each $z$
$$
    \arg \min_{|y|\leq c} |z-y|^2=h_c(z).
$$
Indeed, it is evident that minimizing $y$ has the form $y=\bt z$, where $\bt$, $0\leq
\bt\leq 1$, is constant. Then
$$
    \min_{|y|\leq c} |z-y|^2=\min_{\bt\leq \frac{c}{|z|}}
        (1-\bt)^2|z|^2.
$$

Thus we have to find
$$
    \arg \min_{\bt\leq \frac{c}{|z|}} (1-\bt)^2=
        \min\Big(1,\frac{c}{|z|}\Big).
$$

But last relation is trivially satisfied. Hence the minimizing  $y^*\!=\!z\min(1,\frac{c}{|z|})$ and 
$$
    \arg \min_{|\psi|\leq c} |\psi-A\dot{a}|^2=h_c(A\dot{a}).
$$

From the other side,
$$
    |h_c(z)|^2=|z|^2I_{\{|z|\leq c\}}+
        \frac{|z|^2}{|z|^2}\,c^2\,I_{(|z|\geq c)}\leq c^2.
$$
Hence
$$
    |h_c(z)|\leq c \quad \text{for all} \;\;\;z
$$
and therefore $h_c(A\dot{a})$ satisfies the condition  (\ref{4-2.26}) for each~$A$.

Now it is evident that a function $h_c(A\dot{a})$ minimizes the expression under integral sign, and hence the integral itself over all functions $\psi\in \Psi_0$ satisfying (\ref{4-2.26}).

At the same time the condition (\ref{4-2.25}), generally speaking, can be violated. But, since a matrix $A$ is arbitrary, we can choose $A=A_c^*(\al)$ from (\ref{4-2.28}) which, of course, guarantees the validity of (\ref{4-2.25}) for $\psi_c^*=\psi_c^{A_c^*(\al)}$.
\end{proof}

As we have seen the resulting optimal influence functions  $\psi_c^*$ is defined along the process $Y^0(\al)=(Y_s^0(\al))_{0\leq s\leq t}$, which is a solution of equation (\ref{4-2.2}).

But for constructing optimal estimate we need a function $\psi_c^*(s,x;\al)$ defined on whole space $[0,t]\times C_t\times \cA$.

For this purpose define $\psi_c^*(s,x;\al)$ as follows;
%   (2.29)
\begin{equation}\label{4-2.29}
    \psi_c^*(s,x;\al) =\psi_c^{A_c^*(\al)} (s,x;\al) =
        h_\ve(A_c^*(\al) \dot{a}(s,x;\al)),
\end{equation}
and as usual $\psi_{c,\ve}^*(s,x;\al)=\frac{1}{\ve} \psi_c^*(s,\ve x;\al)$, $0\leq s\leq t$, $x\in C_t$, $\al\in \cA$.

%   Definition 2.2
\begin{definition}\label{d4-2.2}
We say that $\psi_c^*(s,x;\al)$, $0\leq s\leq t$, $x\in C_t$, $\al\in \cA$, is an influence function of optimal $B$-robust estimate $(\al_t^{*,\ve})_{\ve>0}=(\al_t^{\psi_c^{*},\ve})_{\ve>0}$ over the class of CULAN estimates $(\al_t^{\psi,\ve})_{\ve>0}$, $\psi\in \Psi_{0,c}$, if the matrix $A^*(\al)$ is differentiable in $\al$.
\end{definition}

From (\ref{4-2.9}), (\ref{4-2.11}), (\ref{4-2.28}) and (\ref{4-2.29}) it directly follows that 
$$  
    \gm_0^{\psi_c^*}(\al) =P_\al^\ve -\lim_{\ve\\to 0} \ve^2 \gm_\ve^{\psi_c^*}(\al) =
        \int_0^t \psi_c^*(s,Y^0(\al);\al) (\dot{a}(s,Y^0(\al);\al))'ds=Id.
$$
Besides, for each alternative $(P_\al^{\ve,h})_{\ve>0}$, $h\in \cH$, according to the Proposition \ref{p4-2.2} we have 
$$  
    \cL \left\{\ve^{-1} (\al_t^{*,\ve}-\al) \mid P_\al^{\ve,h}\right\} \str{w}{\to} 
        N(b(\psi_c^*,h;\al), V(\psi_c^*;\al)) \;\;\;\text{as} \;\;\; \ve\to 0,
$$
where 
$$  
    b(\psi_c^*,h;\al)=\int_0^t \psi_c^*(s,Y^0(\al);\al) h(s,Y^0(\al);\al)\,ds,
$$
and $V(\psi_c^*;\al)=\Gm_0^{\psi_c^*}(\al)$.

Hence, the risk functional for estimate $(\al_t^{*,\ve})_{\ve>0}$ is 
$$  
    D(\psi_c^*,h;\al)=|b(\psi_c^*,h;\al)|^2 +\tr \Gm_0^{\psi_c^*}, 
        \quad h\in \cH,
$$
and the (unstandardized) gross error sensitivity of $(\al_t^{*,\ve})_{\ve>0}$ is 
$$  
    \gm_{\psi_c^*}(\al) =\sup_{0\leq s\leq t} |\psi_c^*(s,Y^0(\al);\al)|\leq c.
$$

From above reasons, we may conclude that $(\al_t^{*,\ve})_{\ve>0}$ is the optimal $B$-robust estimate over the class of estimates $(\al_t^{\psi,\ve})_{\ve>0}$, $\psi\in \Psi_{0,c}$ in the following sense: the trace of asymptotic covariance matrix of $(\al_t^{*,\ve})_{\ve>0}$ is minimal among all estimates $(\al_t^{\psi,\ve})_{\ve>0}$ with bounded by constant gross error sensitivity, that is
$$  
    \Gm_0^{\psi_c^*}(\al) \leq \Gm_0^\psi(\al) \;\;\;\text{for all} \;\;\; 
        \psi\in \Psi_{0,c}\,.   \;\;\qed
$$

Note that for each estimate $(\al_t^{\psi,\ve})_{\ve>0}$ and alternatives $(P_\al^{\ve,h})_{\ve>0}$, $h\in \cH$, the influence functional is bounded by $const(r) \cdot c$. Indeed, we have for $\psi\in \Psi_{0,c}$,
$$  
    \sup_{h\in \cH_r} |b(\psi,h;\al)|\leq const(r) \cdot c:= C(r,c),
$$
and since from (\ref{4-2.24})
$$  
    \inf_{\psi\in \Psi_{0,c}} \sup_{h\in \cH_r} D(\psi,h;\al) \leq 
        C^2(r,c)+\tr \Gm_0^{\psi_c^*}(\al),
$$
we can choose ``optimal level'' of truncation, minimizing the expression 
$$  
    C^2(r,c)+\tr \Gm_0^{\psi_c^*}(\al)
$$
over all constants $c$, for which the equation (\ref{4-2.28}) has a solution $A_c^*(\al)$. This can be done using the numerical methods.

For the problem of existence and uniqueness of solution of equation (\ref{4-2.28}) we address to Rieder \cite{35}. 

In the case of one-dimensional parameter $\al$ (i.e. $m=1$) the optimal level $c^*$ of truncation is given as an unique solution of the following equation (see Lazrieva and Toronjadze \cite{21}, \cite{22})
$$  
    r^2c^2 =\int_0^t [\dot{a} (s,Y^0(\al);\al)]_{-c}^c \dot{a}(s,Y^0(\al);\al)\,ds-
        \int_0^t ([\dot{a}(s,Y^0(\al);\al)]_{-c}^c)^2 \,ds,
$$
where $[x]_a^b=(x\wedge b)\vee a$ and the resulting function 
$$  
    \psi^*(s,x;\al) =[\dot{a}(s,x;\al)]_{-c^*}^{c^*} , \quad 
        0\leq s\leq t, \quad x\in C_t,
$$
is $(\Psi_0,\cH_r)$ optimal in the following minimax sense:
$$  
    \sup_{h\in \cH_r} D(\psi^*,h;\al) =\inf_{\psi\in \Psi} 
        \sup_{h\in \cH_r} D(\psi,h;\al).
$$

%   3
\section{Optimal Mean-Variance Robust Hedging}

\subsection{A financial market model}
Let $(\Om,\cF,F=(\cF_t)_{0\leq t\leq T},P)$ be a filtered probability space with filtration $F$ satisfying the usual conditions, where $T\in (0,\infty]$ is a fixed time horizon. Assume that $\cF_0$ is a trivial and $\cF_T=\cF$.

There exist $d+1$, $d\geq 1$ primitive assets: one bound, whose price process is assumed to be 1 at all times and $d$ risky assets (stocks), whose $R^d$-valued price process $X=(X_t)_{0\leq t\leq T}$ is a continuous semimartingale given by the relation: 
%   (3.1)
\begin{equation}\label{4-3.1}
    dX_t =\diag (X_t) \,dR_t, \quad X_0>0,
\end{equation}
where $\diag(X)$ denotes the diagonal $d\times d$-matrix with diagonal elements $X^1,\dots, X^d$, and the yield process $R=(R_t)_{0\leq t\leq T}$ is a $R^d$-valued continuous semimartingale satisfying the stricture condition (SC). That is (see Schweizer \cite{38}) 
%   (3.2)
\begin{equation}\label{4-3.2}
    d R_t=d\la \wt M\ra_t\lb_t+d\wt M_t, \quad R_0=0,
\end{equation}
where $\wt M=(\wt M_t)_{0\leq t\leq T}$ is a $R^d$-valued continuous martingale, $\wt M\in \cM_{0,\loc}^2(P)$, $\lb=(\lb_t)_{0\leq t\leq T}$ is a $F$-predictable $R^d$-valued process, and the mean-variance tradeoff (MVT) process $\wt \cK=(\wt\cK_t)_{0\leq t\leq T}$ of process $R$
%   (3.3)
\begin{equation}\label{4-3.3}
    \wt\cK_t:=\int_0^t \lb_s' d\la \wt M\ra_s\lb_s =\la \lb'\cdot \wt M\ra_t<\infty
        \;\;\Pas, \;\; t\in [0,T].
\end{equation}

%   Remark 3.1
\begin{remark}\label{r4-3.1}
Remember that all vectors are assumed to be column vectors.
\end{remark}

Suppose that the martingale $\wt M$ has the form 
%   (3.4)
\begin{equation}\label{4-3.4}
    \wt M=\sg\cdot M,
\end{equation}
where $M=(M_t)_{0\leq t\leq T}$ is a $R^d$-valued continuous martingale, $M\in \cM_{0,\loc}^2(P)$ with $d\la M^i,M^j\ra_t=I_{i_j}^{d\times d} d C_t$, $I^{d\times d}$ is the identity matrix, $C=(C_t)_{0\leq t\leq T}$ is a continuous increasing bounded process with $C_0=0$.

Further, let $\sg=(\sg_t)_{0\leq t\leq T}$ is a $d\times d$-matrix valued, $F$-predictable process with $\rank(\sg_t)=d$ for any $t$, $P$-a.s., the process $(\sg_t^{-1})_{0\leq t\leq T}$ is locally bounded, and 
%   (3.5)
\begin{equation}\label{4-3.5}
    \int_0^T \sg_t\,d\la M\ra_t \sg_t'<\infty \;\; \Pas
\end{equation}

Assume now that the following condition be satisfied:

There exist fixed $R^d$-valued, $F$-predictable process $k=(k_t)_{0\leq t\leq T}$ such that 
%   (3.6)
\begin{equation}\label{4-3.6}
    \lb=\lb(\sg)=(\sg')^{-1}k.
\end{equation}

In the case from (\ref{4-3.2}) we get 
%   (3.7)
\begin{gather}
    dR_t=d\la \wt M\ra_t \lb_t +d\wt M_t =\sg_td\la M\ra_t \sg_t'
        (\sg_t')^{-1} k_t+\sg_t d M_t \notag \\
    =\sg_t(d\la M\ra_tk_t+dM_t), \label{4-3.7}
\end{gather}
and 
\begin{gather*}
    \wt\cK_t =\int_0^t \lb_s' d\la \wt M\ra_s\lb_s =
        \int_0^t k_t'((\sg_t')^{-1})' \sg_t d\la \wt M\ra_t 
        \sg_t'(\sg_t')^{-1} k_t \\
    =\int_0^t k_t' d\la M\ra_t k_t=\la k\cdot M\ra_t:=\cK_t.
\end{gather*}

From (\ref{4-3.3}) we have 
%   (3.8)
\begin{equation}\label{4-3.8}
    \cK_t<\infty \;\; \Pas \;\;\text{for all} \;\; t\in [0,T].
\end{equation}

Thus, of we introduce the process $M^0=(M_t^0)_{0\leq t\leq T}$ by the relarion 
%   (3.9)
\begin{equation}\label{4-3.9}
    d M_t^0=d\la M\ra_t k_t +d M_t, \quad M_0^0=0,
\end{equation}
then the MVT process $\cK=(\cK_t)_{0\leq t\leq T}$ of $R^d$-valued semimartingale $M^0$ is finite, and hence $M^0$ satisfies SC.

Finally, the scheme (\ref{4-3.1}), (\ref{4-3.2}), (\ref{4-3.4}), (\ref{4-3.6}) and (\ref{4-3.9}) can be rewritten in the following form 
%   (3.10)
\begin{equation}\label{4-3.10}
\begin{aligned}
    dX_t & =\diag (X_t) \,dR_t, \quad X_0>0, \\
    dR_t & =\sg_t\,d M_t^0, \quad R_0=0, \\
    d M_t^0 & =d\la M\ra_t k_t+d M_t, \quad M_0=0,
\end{aligned}
\end{equation}
where $\sg$ and $k$ satisfy (\ref{4-3.5}) and (\ref{4-3.8}), respectively.

This is our financial market model.

%   3.2
\subsection{Characterization of variance-optimal ELMM (equivalent local martingale measure)}

A key role in mean-variance hedging plays variance-optimal ELMM (see, e.g., RSch \cite{34}, GLP \cite{11}). Here we collect some facts characterizing this measure.

 We start with remark that the sets ELMMs for processes $X$, $R$ and $M^0$ form (\ref{4-3.10}) coincide. Hence we can and will consider the simplest process $M^0$. 

Introduce the notation 
$$  
    \cM_2^e:= \left\{ Q\sim P\,:\, \frac{dQ}{dP}\in L^2(P), \;\; M^0 \;\;
        \text{is a $Q$-local martingale}\right\},
$$
and suppose that 
$$  
    (c.1) \hskip+3cm \cM_2^e \neq \vnth. \hskip+5cm 
$$
The solution $\wt P$ of the optimization problem
$$  
    E\cE_T^2(\cM^Q)\to \inf_{Q\in \cM_2^e}
$$
is called variance-optimal ELMM.

Here 
$$  
    \frac{dQ}{dP}\Big|_{\cF_T} =\cE_T(M^Q),
$$
and $(\cE_t(M^Q))_{0\leq t\leq T}$ is the Dolean exponential of martingale $M^Q$. 

It is well-known (see, e.g., Schweizer \cite{38}, \cite{39}) that under condition (c.1) variance-optimal ELMM $\wt P$ exist.

Denote 
$$  
    \wt z_T:=\frac{d\wt P}{dP} \Big|_{\cF_T} \,,
$$
and introduce RCLL process $\wt z=(\wt z_t)_{0\leq t\leq T}$ by the relation 
$$  
    \wt z_t=E^{\wt P} (\wt z_T/\cF_T), \quad 0\leq t\leq T.
$$

Then, by Schweizer \cite{38}, \cite{39}
%   (3.11)
\begin{equation}\label{4-3.11}
    \wt z_T=\wt z_0+\int_0^T \zt_t'\,d M_t^0,
\end{equation}
where $\zt=(\zt_t)_{0\leq t\leq T}$ is the $R^d$-valued $F$-predictable process with 
$$  
    \int_0^T \zt_t'\,d\la M\ra_t \zt_t<\infty,
$$
and the process $\big(\int\limits_0^t \zt_s' dM_s^0\big)_{0\leq t\leq T}$ is a $\wt P$-martingale.

Relation (\ref{4-3.11}) easily implies that the process $\wt z$ is actually continuous.

Suppose, in addition to (c.1), that the following condition is satisfied:

\smallskip

(c.$^*$) all $P$-local martingales are continuous.
\smallskip

This technical assumption is satisfied in stochastic volatility models, where $F=F^w$ is the natural filtration generated by the Wiener process.

It shown in Mania and Tevzadze \cite{35}, Mania et al. \cite{26} that under conditions (c.1) and (c$^*$) density $\wt z_T$ of variance optimal ELMM is uniquely characterized by the relation
%   (3.12)
\begin{equation}\label{4-3.12}
    \wt z_T =\frac{\cE_T((\vf-k)'\cdot M^0)}{E\cE_T((\vf-k)'\cdot M^0)}\,,
\end{equation}
where $\vf$ together with the pair $(L,c)$ is the unique solution of the following equation 
%   (3.13)
\begin{equation}\label{4-3.13}
    \frac{\cE_T((\vf-2k)'\cdot M)}{\cE_T(L)} =c\cE_T^2(-k'\cdot M),
\end{equation}
where $L\in M_{0,\loc}^2(P)$, $\la L,M\ra=0$, $c$ is a constant.

Moreover, the process $\zt=(\zt_t)_{0\leq t\leq T}$ from (\ref{4-3.11}) has the form 
%   (3.14)
\begin{equation}\label{4-3.14}
\zt_t=(\vf_t-k_t)\cE_t((\vf-k)'\cdot M^0).
\end{equation}

Here $\vf=(\vf_t)_{0\leq t\leq T}$ is a $R^d$-valued, $F$-predictable process with 
$$  
    \int_0^T \vf_t'\,d\la M\ra_t\vf_t<\infty.
$$

Let $\tau$ be $F$-stopping time.

Denote $\la k'\cdot M\ra_{T\tau}=\la k'\cdot M\ra_T-\la k'\cdot M\ra_\tau$.

%   Proposition 3.1
\begin{proposition}[{see also Biagini et al. \cite{3}, LLaurent and Pham \cite{19}}] \label{p4-3.1}
\

$1.$ Equation $(\ref{4-3.13})$ is equivalent to equation 
%   (3.15)
\begin{equation}\label{4-3.15}
    \frac{\cE_T(\vf'\cdot M^*)}{\cE_T(L)} =ce^{\la k'\cdot M\ra_T},
\end{equation}
where the $R^d$-valued process $M^*=(M_t^*)_{0\leq t\leq T}$ is given by the relation 
$$  
    d M_t^*=2d\la M\ra_t k_t +dM_t, \quad M_0^*=0.
$$

$2.$ {\rm a)} If there exists the martingale $m=(m_t)_{0\leq t\leq T}$, $m\in \cM_{0,\loc}^2(P)$ such that 
%   (3.16)
\begin{equation}\label{4-3.16}
    e^{-\la k'\cdot M\ra_T} =c+m_T, \quad \la m,M\ra=0,
\end{equation}
then $\vf\equiv 0$ and $L_T=\int\limits_0^T \frac{1}{c+m}\,dm_s$ solve the equation $(\ref{4-3.15})$.

In this case 
%   (3.17)
\begin{equation}\label{4-3.17}
    \wt z_T=\frac{\cE_T(-k'\cdot M^0)}{E\cE_T(-k'\cdot M^0)}\,,
\end{equation}
process $\zt=(\zt_t)_{0\leq t\leq T}$ from $(\ref{4-3.11})$ is equal to 
$$  
    \zt_t=-k_t\cE_t(-k'\cdot M^0),
$$
and
$$  
    E\left[ \left( \frac{\wt z_T}{\wt z_\tau}\right)^2 \Big/ \cF_\tau\right] =
        \frac{1}{E(e^{-\la k'\cdot M\ra_{T\tau}} /\cF_\tau)}\,.
$$

{\rm b)} If there exist $R^d$-valued $F$-predictable process $\ell=(\ell_t)_{0\leq t\leq T}$, $\int\limits_0^T \ell_t'd \la M\ra \ell_t <\infty$ and 
$$  
    e^{\la k'\cdot M\ra_T}=c+\int_0^T \ell_t'\,d M_t^*,
$$
then $L\equiv 0$ and $\vf_t=\frac{\ell_t}{c+\int_0^t \ell_s'dM_s^*}$ solve the equation $(\ref{4-3.15})$.

In this case 
$$
    \wt z_T=\cE_T(-k'\cdot M) \;\; (:=\wh z_T, \;\;
        \text{the density of minimal martingale measure} \;\; \wh P),
$$
and
$$  
    E\left(\left( \frac{\wt z_T}{\wt z_\tau}\right)^2 \Big/ \cF_\tau\right) =
        E^{P^*} (e^{\la k'\cdot M\ra_{T\tau}} \big/ \cF_\tau),
$$
where $dP^*=\cE_T(-2k'\cdot M) dP$.
\end{proposition}

\begin{proof}
1. By the Yor formula 
\begin{gather*}
    \cE_T\left(\vf-2k)'\cdot M\right)=
        \cE_T(\psi'\cdot M-2k'\cdot M) \\
    =\cE_T\Bigg( \vf'\cdot \bigg( M+2\int_0^\cdot d\la M\ra_t k_t \bigg) -
        2\int_0^\cdot \psi_t' d\la M\ra_t k_t-2k'\cdot M\Bigg) \\
    =\cE_T(\vf'\cdot M^*)\cE_T(-2k'\cdot M),
\end{gather*}
and 
$$  
    \cE_T^2(-k'\cdot M)=\cE_T(-2k'\cdot M) e^{\la k'\cdot M\ra_T}.
$$
Assertion follows.

2. a) Note at first that $\la L,M\ra=0$. Further, by the formula we can write 
$$  
    \ln(c+m_t)-\ln c=\int_0^t \frac{1}{c+m_s}\,d m_s -
        \frac{1}{2} \int_0^t \frac{1}{(c+m_s)^2}\, d\la m\ra_s.
$$

Hence 
$$  
    e^{\ln(c+m_T)-\ln c} =\cE_T(L),
$$
and thus 
$$  
    \cE_T(L) =\frac{c+m_T}{c} =\frac{e^{-\la k'\cdot M\ra_T}}{c}\,.
$$

Finally, by the Bayes rule and the Girasnov Theorem 
\begin{gather*}
    E\left( \left( \frac{\wt z_T}{\wt z_\tau}\right)^2 \Big/\cF_\tau\right)=
        \frac{E(\cE_T(-2k'\cdot M)e^{-\la k'\cdot M\ra_T}/\cF_\tau)}
            {E^2(\cE_T(-k'\cdot M)e^{-\la k'\cdot M\ra_T}/\cF_\tau)}\\
    =\frac{E^*(c+m_T/\cF_\tau)\cE_T^2(-k'\cdot M)}
        {(\wh E(c+m_\tau/\cF_\tau))^2\cE_T^2(-2k'\cdot M)} =
        \frac{c+m_\tau}{(c+m_\tau)^2}\cdot e^{\la k'\cdot M\ra_\tau} \\
    =\frac{1}{E(e^\la k'\cdot M\ra_{T\tau}/\cF_\tau)}\,.
\end{gather*}

The proof of case 2 b) is quite analogous. 

Proposition is proved.
\end{proof}

%   3.3
\subsection{Misspecified asset price model and robust hedging}               
Denote by $\Ball_L(0,r)$, $r\in [0,\infty)$ the closed $r$-radius ball in the space $L=L_\infty(dt\times dP)$, with the center at the origin, and let 
%   (3.18)
\begin{gather}
    \cH:= \big\{ h=\{h_{ij}\}, \; i,j=\wh{1,d}:h \;\text{is $F$-predictable} \;\;
        d\times d\text{-matrix} \notag \\
    \text{valued process,} \;\; \rank(h)=d, \;\; h_{ij}\in \Ball_L(0,r), \;\; 
        r\in [0,\infty)\big\}.
\label{4-3.18}
\end{gather}

Class $\cH$ is called the class of alternatives.

Fix the value of small parameter $\dl>0$, as well as $d\times d$-matrix valued, $F$-predictable process $\sg^0=(\sg_t^0)_{0\leq t\leq T}= (\{\sg_{ij,t}^0\}, \,1\leq i,j\leq d)_t$ 
such that $|\sg_{ij,t}^0|\leq const$, $\fa i,j,t$, the matrix $(\sg^0)^2=\sg^0(\sg^0)'$ is uniformly elliptic, i.e. for each vector $v_t=(v_t',\dots,v_t^d)$ with probability 1
\begin{equation}\label{1.4.4}
    \sum_{i,j=1}^d (\sg^0)_{ij,t}^2 v_t^i v_t^j \geq c\sum_{i=1}^d |v_t^i|^2,
        \quad c>0, \quad 0\leq t\leq T,
\end{equation}
%$\rank (\sg^0)=d$, with 
%   (3.19)
%\begin{equation}\label{4-3.19}
%    \int_0^T \sg_t^0 d\la M\ra_t (\sg_t^0)'<\infty \;\;\Pas
%\end{equation}
and denote 
%   (3.20)
\begin{equation}\label{4-3.20}
    A_\dl=\{\sg: \sg=\sg^0+\dl h, \;\; h\in \cH\}.
\end{equation}

%   Proposition 1.
\begin{proposition}\label{1.4.1}
Every $\sg$ from the class $A_\dl$  for sufficiently small  $\dl$ is  $F$-predictable $d\times d$-valued process with bounded elements and the matrix $\sg^2=\sg\sg'$ is uniformly elliptic.
\end{proposition}

\begin{proof}
The process $\sg$ is $F$-predictable as linear combination of $F$-predictable processes. Further,
$$
|\sg_{ij,t}|=|\sg_{ij,t}^0+\dl h_{ij,t}|\leq const+\dl r, \quad 
    0<\dl\ll 1.
$$

From (\ref{1.4.4}) and (\ref{4-3.20}) for each vector  $\nu_t=(\nu_t^1,\dots,\nu_t^d)$ we have
%   (1.4.6)
\begin{align}
\sum_{i,j=1}^d(\sg^2)_{ij,t}\nu_t^i\nu_t^j & =
    \sum_{i,j=1}^d(\sg^0+\dl h)(\sg^0+\dl h)_{ij,t}'\nu_t^i\nu_t^j\notag\\
& =\sum_{i,j=1}^d (\sg^0(\sg^0)')_{ij,t}\nu_t^i\nu_t^j+
    \dl\sum_{i,j=1}^d (\sg^0h')_{ij,t}\nu_t^i\nu_t^j \notag \\
& \quad +\dl\sum_{i,j=1}^d (h(\sg^0)')_{ij,t}\nu_t^i\nu_t^j+
    \dl^2\sum_{i,j=1}^d(hh')_{ij,t}\nu_t^i\nu_t^j. \label{1.4.6}
\end{align}

Note now that the elements of matrices $\sg^0$ and $h$ are bounded. Hence choosing  $\dl$ sufficiently small we get 
$$
\max\left(\dl|(\sg^0h')_{ij,t}|, \dl|(h(\sg^0))_{ij,t}|, 
    \dl^2|(hh')_{ij,t}|\right)\leq \frac{\ve}{3}\,.
$$

Therefore from (\ref{1.4.4}) and (\ref{1.4.6}) we get
\begin{gather*}
\sum_{i,j=1}^d \sg_{ij,t}^2 \nu_t^i\nu_t^j\geq 
    (c-const \cdot \ve)\sum_{i,j=1}^d|\nu_t^i|^2 \;\;\;\text{for each $\ve>0$}. 
\end{gather*}

Proposition is proved.
\end{proof}

Consider the set of processes $\{R^\sg (\text{or}\; X^\sg)$, $\sg\in A_\dl\}$, which represents the misspecified of asset price model.

Define the class of admissible trading strategies $\Tht=\Tht(\sg^0)$.

%   Proposition 2. 
\begin{proposition}\label{1.4.2}
For each $R^d$-valued $F$-predictable process  $\tht=(\tht_t)_{0\leq t\leq T}$ and for each  $\sg\in A_\dl$, $\dl>0$,
\begin{align*}
aE\int_0^T|\tht_t|^2dC_t &  \leq E\int_0^T \tht_t'\sg_td\la M\ra_t\sg_t'\tht_t \\
& =E\int_0^T \tht_t'\sg_t\sg_t'\tht_tdC_t\leq 
    AE\int_0^T|\tht_t|^2dC_t,
\end{align*}
where the constants $a$, $A$ are such that  $0<a\leq A<\infty$, and the parameter $\dl>0$ is sufficiently small.
\end{proposition}

\begin{proof}
Remember that   $d\la M\ra_t=d\la M^i,M^j\ra_t=I_{ij}^{d\times d} dC_t$. Hence 
$$
E\int_0^T \tht_t'\sg_td\la M\ra_t\sg_t'\tht_t=
    E\int_0^T \tht_t'\sg_t\sg_t'\tht_t\,dC_t.
$$

Further, since $\sg=\sg^0+\dl h$ and elements of matrices  $\sg^0$ and $h$ are bounded, then the same is true for the elements of matrix $\sg$ with $0\leq \dl\leq const$. Thus using the inequality $ab\leq 2(a^2+b^2)$ we get 
$$
E\int_0^T \tht_t'\sg_t\sg_t'\tht_t\,dC_t\leq AE\int_0^T|\tht_t|^2dC_t.
$$
On the other hand, by Proposition \ref{1.4.1} the matrix $\sg^2=\sg\sg'$ is uniformly elliptic for sufficiently small $\dl$, which yields the first inequality.  
\end{proof}

%   Definition 3.1
\begin{definition}\label{d4-3.1}
The class $\Tht=\Tht(\sg^0)$ is a class of $R^d$-valued $F$-predictable processes $\tht=(\tht_t)_{0\leq t\leq T}$ such that 
%   (3.21)
\begin{equation}\label{4-3.21}
    E\int_0^T |\tht_t|^2 d C_t<\infty.
\end{equation}
\end{definition}

Let $\tht\in \Tht$ be the dollar amount (rather than the number of shares) invested in the stock $X^\sg$, $\sg\in A_\dl$. Then for each $\sg\in A_\dl$ the trading gains induced by the self-financing portfolio strategy associated to $\tht$ has the form 
%   (3.22)
\begin{equation}\label{4-3.22}
    G_t(\sg,\tht)=\int_0^t \tht_s'\,d R_s^\sg, \quad 0\leq t\leq T,
\end{equation}
where $R^d=(R_t^d)_{0\leq t\leq T}$ is the yield process given by (\ref{4-3.10}).

Introduce the condition: 
\smallskip

(c.2) There exists ELMM $\ol{Q}$ such that the density process $z=z^{\ol{Q}}$ satisfies the reverse H\"{o}lder inequality $R_2(P)$, see definition in RSch \cite{34}.
\smallskip

It is well-known that under the conditions (c.1) and (c.2) the density process $\wt z=(\wt z_t)_{\leq t\leq T}$ of the variance-optimal ELMM satisfies $R_2(P)$ as well, see Dolean et al. \cite{8}.

Now under the conditions (c.1) and (c.2) the r.v. $G_T(\sg,\tht)\in L^2(P)$, $\fa \sg\in A_\dl$, and the space $G_T(\sg,\Tht)$ is closed in $L^2(P)$, $\fa \sg\in A_\dl$ (see, e.g., Theorem 2 of RSch \cite{34}).

A contingent claim is an $\cF_T$-measurable square-integrable r.v. $H$, which models the payoff from a financial product at the maturity date $T$. 

The problem we are interested in is to find the robust hedging strategy for a contingent claim $H$ in the above described incomplete financial market model with misspecified asset price process $X^\sg$, $\sg\in A_\dl$, using mean-variance approach. 

For each $\sg\in A_\dl$, the total loss of a hedger, who starts with the initial capital $x$, uses the strategy $\tht$, believes that the stock price process follows $X^\sg$, and has to pay a random amount $H$ at the date $T$, is $H$-$x$-$G_T(\sg,\tht)$.

Denote 
%   (3.23)
\begin{equation}\label{4-3.23}
    \cJ(\sg,\tht):=E(H-x-G_T(\sg,\tht))^2.
\end{equation}

One setting of the robust mean-variance hedging problem consist in solving the optimization problem 
%   (3.24)
\begin{equation}\label{4-3.24}
    \text{minimize} \;\; \sup_{\sg\in A_\dl} \cJ(\sg,\tht) \;\;
        \text{over all strategies} \;\; \tht\in \Tht.
\end{equation}

We ``slightly'' change this problem using the approach developed in Toronjadze \cite{42} which based on the following approximation 
\begin{align*}
    \sup_{\sg\in A_\dl} \cJ(\sg,\tht) & =\exp \big\{ \sup_{h\in \cH} 
        \ln\cJ(\sg^0+\dl h,\tht)\big\} \\
    & \simeq \exp \bigg\{ \sup_{h\in \cH} \bigg[ \ln \cJ(\sg^0,\tht) +
        \dl\,\frac{D\cJ(\sg^0,h,\tht)}{\cJ(\sg^0,\tht)}\bigg]\bigg\} \\
    & =\cJ(\sg^0,\tht) \exp \bigg\{ \dl\sup_{h\in \cH} 
        \frac{D\cJ(\sg^0,h,\tht)}{\cJ(\sg^0,\tht)}\bigg\},
\end{align*}
where 
$$  
    D\cJ(\sg^0,h,\tht):= \frac{d}{d\dl}\,\cJ(\sg^0+\dl h,\tht) |_{\dl=0} =
        \lim_{\dl\to 0} \frac{\cJ(\sg^0+\dl h,\tht)-\cJ(\sg^0,\tht)}{\dl}\,,
$$
is the Gateaux differential of the functional $\cJ$ at the point $\sg^0$ in the direction~$h$. 

Approximate (in leading order $\dl$) the optimization problem (\ref{4-3.24}) by the problem 
%   (3.25)
\begin{gather}
    \text{minimize} \;\; \cJ(\sg^0,\tht) \exp \bigg\{ \dl\sup_{h\in \cH} 
        \frac{D\cJ(\sg^0,h,\tht)}{\cJ(\sg^0,\tht)}\bigg\} \notag \\
    \text{over all strategies} \;\; \tht\in \Tht.  \label{4-3.25}
\end{gather}

Note that each solution $\tht^*$ of the problem (\ref{4-3.25}) minimizes $\cJ(\sg^0,\tht)$ under the constraint 
$$  
    \sup_{h\in \cH} \frac{D\cJ(\sg^0,h,\tht)}{\cJ(\sg^0,\tht)} \leq c:=
        \sup_{h\in \cH} \frac{D\cJ(\sg^0,h,\tht^*)}{\cJ(\sg^0,\tht^*)}\,.
$$

This characterization of an optimal strategy $\tht^*$ of the problem (\ref{4-3.25}) leads to the 
%   Definition 3.2
\begin{definition}\label{d4-3.2}
The trading strategy $\tht^*\in \Tht$ is called optimal mean-variance robust trading strategy against the class of alternatives $\cH$ if it is a solution of the optimization problem 
%   (3.26)
\begin{gather}
    \text{minimize} \;\; \cJ(\sg^0,\tht)\;\;\text{over all strategies} \;\; 
        \tht\in \Tht, \;\; \text{subject to constraint} \notag \\
    \sup_{h\in \cH} \frac{D\cJ(\sg^0,h,\tht)}{\cJ(\sg^0,\tht)}\leq c, \label{4-3.26}
\end{gather}
where $c$ is some generic constant.
\end{definition}

%   Remark 3.3
\begin{remark}\label{r4-3.3}
In contrast to ``mean-variance robust'' trading strategy which associates with optimization problem (\ref{4-3.24}) and control theory, we find the ``optimal mean-variance robust'' strategy in the sense of Definition \ref{d4-3.2}. Such approach and term are common in robust statistics theory (see, e.g., Hampel et al. \cite{12}, Rieder \cite{35}).
\end{remark}

Does the suggested approach provide ``good'' approximation?
Consider the case.

\smallskip

{\bf Diffusion model with zero drift.}
Let a standard Wiener process $w=(w_t)_{0\leq t\leq T}$ be given 
on the complete probability space $(\Om,\cF,P)$. 
Denote by $F^w=(\cF_t^w$, $0\leq t\leq T)$ the $P$-augmentation of the natural
filtration $\cF_t^w=\sg(w_s$, $0\leq s\leq t)$, $0\leq t\leq T$, 
generated by $w$.

    Let the stock price process be modeled by the equation 
%   (7)
\begin{equation*} %\label{ch7}
    dX_t^\sg=X_t^\sg\cdot \sg_t\,dw_t, \;\; X_0^\sg>0, \;\; 0\leq t\leq T,
\end{equation*}
where $\sg\in A_\dl$ with 
$$  \int\limits_0^T (\sg_t^0)^2\,dt<\infty     $$ 
and $h\in \Ball_{L_\infty(dt\times dP)}(0,r)$, $0<r<\infty$.
All considered processes are real-valued.

    Denote by $R^\sg$ the yield process, i.e.,
%   (8)
\begin{equation*} %\label{ch8}
    dR_t^\sg=\sg_t\,dw_t, \;\; R_0^\sg=0, \;\; 0\leq t\leq T.
\end{equation*}

The wealth at maturity $T$, with the initial endowment $x$, is equal to 
$$  V_T^{x,\tht}(\sg)=x+\int_0^T \tht_t\,dR_t^\sg.      $$
Let, further, the contingent claim $H$ be $\cF_T^w$-measurable 
$P$-square-integrable r.v.

    Consider the optimization problem (\ref{4-3.24}). 
It is easy to see that if $\sg\in A_\dl$; then
$$  \sg_t^0-\dl r\leq \sg_t\leq \sg_t^0+\dl r, \;\; 0\leq t\leq T, \;\; 
        \Pas,       $$

    By the martingale representation theorem
%   (11)
\begin{equation*} %\label{ch11}
    H=EH+\int_0^T \vf_t^H\,dw_t, \;\; \Pas,
\end{equation*}
where $\vf^H$ is the $F^w$-predictable process with
%   (12)
\begin{equation}\label{ch12}
    E\int_0^T (\vf_t^H)^2\,dt<\infty.
\end{equation}
Hence
$$  E\big(H-V_T^{x,\tht}(\sg)\big)^2=
        (EH-x)^2+E\int_0^T (\vf_t^H-\sg_t\tht_t)^2\,dt.     $$
From this it directly follows that the process
%   (13)
\begin{align}
    \sg_t^{*}(\tht) & =(\sg_t^0-\dl r)
        I_{\{\frac{\vf_t^H}{\tht_t}\geq \sg_t^0\}}
            I_{\{\tht_t\neq 0\}} \notag \\
    & +(\sg_t^0+\dl r)I_{\{\frac{\vf_t^H}{\tht_t}<\sg_t^0\}}
        I_{\{\tht_t\neq 0\}}, \;\; 0\leq t\leq T,\label{ch13}
\end{align}
is a solution of the optimization problem
$$  \text{maximize} \;\; E\big(H-V_T^{x,\tht}(\sg)\big)^2 \;\; 
        \text{over all} \;\; \sg\in A_\dl, \;\; 
            \text{with a given} \;\; \tht\in \Tht.      $$
It remains to minimize (w.r.t. $\tht$) the expression
$$  E\int_0^T \big(\vf_t^H-\sg_t^{*}(\tht)\tht_t\big)^2\,dt.    $$
From (\ref{ch13}) it easily follows that the equation (w.r.t. $\tht$)
$$  \vf_t^H-\sg_t^{*}(\tht)\tht_t=0,    $$
has no solution, but
%   (14)
\begin{equation}\label{ch14}
    \tht_t^{*}=\frac{\vf_t^H}{\sg_t^0}\,I_{\{\sg_t^0\neq 0\}}, \;\; 
        0\leq t\leq T,
\end{equation}
solves problem. We assume that $0/0:=0$. 

    Consider now the optimization problem (\ref{4-3.26}).

    For each fixed $h$
\allowdisplaybreaks
\begin{align*}
    J(\sg,\tht) & =E\Big(H-x-\int_0^T \tht_t\,dR_t^\sg\Big)^2 \\
    & =E\bigg(H-x-\int_0^T \tht_t\sg_t^0\,dw_t-
        \dl\int_0^T \tht_th_t\,dw_t\bigg)^2 \\
    & =J(\sg^0,\tht)-2\dl E\bigg[\Big(EH-x+
            \int_0^T \big(\vf_t^H-\tht_t\sg_t^0\big)\,dw_t\Big)
        \int_0^T \tht_th_t\,dw_t\bigg] \notag \\
    & \quad +\dl^2E\int_0^T \tht_t^2h_t^2\,dt,
\end{align*}
and hence
%   (15)
\begin{equation}\label{ch15}
    DJ(\sg^0,h;\tht)=
        2E\int_0^T \big(\tht_t\sg_t^0-\vf_t^H\big)\tht_th_t\,dt,
\end{equation}
as follows from (\ref{ch12}), the definition of the class $\cH$ and 
the estimation
%   (16)
\begin{gather}
    \Big(E\int_0^T \big(\tht_t\sg_t^0-\vf_t^H\big)\tht_th_t\,dt\Big)^2\leq
            E\int_0^T \big(\tht_t\sg_t^0-\tht_t^H\big)^2\,dt\;
        E\int_0^T \tht_t^2h_t^2\,dt \notag \\
    \leq \const\cdot r^2\bigg(E\int_0^T \tht_t^2(\sg_t^0)^2\,dt+
        E\int_0^T (\vf_t^H)^2\,dt\bigg)E\int_0^T \tht_t^2\,dt<\infty.\label{ch16}
\end{gather}
Since, further, $DJ(\sg^0,h;\tht)=0$ for $h\equiv 0$, using (\ref{ch16}) we get
$$  0\leq \sup_{h\in \cH} DJ(\sg^0,h;\tht)<\infty.  $$
Hence we can take $0\leq c<\infty$ in problem (6). 
Now if we substitute $\tht^{*}$ from (\ref{ch14}) into (\ref{ch15}), 
we get $DJ(\sg^0,h;\tht^{*})=0$ for each $h$, and thus
$$  \frac{\sup\limits_{h\in \cH} 
        DJ(\sg^0,h;\tht^{*})}{J(\sg^0,\tht^{*})}=0. $$
If we recall that 
$\tht^{*}=\arg\min\limits_{\tht\in \Tht_{A_\dl}} J(\sg^0,\tht)$, 
we get that $\tht^{*}$ defined by (\ref{ch14}) is a solution 
of this optimization problem  as well.

    Thus we prove that

{\em    {\rm (a)} the mean-variance robust trading strategy 
$\tht^{*}=(\tht_t^{*})_{0\leq t\leq T}$ for the optimization problem $(\ref{4-3.24})$ 
is given by the formula
$$  \tht_t^{*}=\frac{\vf_t^H}{\sg_t^0}\,I_{\{\sg_t^0\neq 0\}};  $$

    {\rm (b)} at the same time this strategy is an optimal mean-variance robust trading 
strategy for the optimization problem $(\ref{4-3.26})$. 

Hence in this case the suggested approach leads to the perfect solution of initial problem $(\ref{4-3.24})$.
}

\smallskip

To solve the problem (\ref{4-3.26}) in general case we need to calculate $D\cJ(\sg^0,h,\tht)$.
Suppose that $k=(k_t)_{0\leq t\leq T}=(k_{i,t},\, 1\leq i\leq d)_{0\leq t\leq T}$ from (\ref{4-3.10}) is such that $|k_{i,t}|\leq const$ $\fa i,t$.

Following RSch \cite{34} and GLP \cite{11} introduce the probability measure $\wt Q\sim P$ on $\cF_T$ by the relation 
%   (3.27)
\begin{equation}\label{4-3.27}
    d\wt Q=\frac{\wt z_T}{\wt z_0} \,d\wt P \;\;\;\Big(\text{and hence} \;\; 
        d\wt Q=\frac{\wt z_T^2}{\wt z_0} \, dP\Big).
\end{equation}
Using Proposition 5.1 of GLP \cite{11} we can write 
%   (3.28)
\allowdisplaybreaks
\begin{gather}
    \cJ(\sg,\tht)= E\,\frac{\wt z_T^2}{\wt z_0^2}\,
        \frac{\wt z_0^2}{\wt z_T^2}\, \Bigg( H-x-\int_0^T \tht_t'\,d R_t^\sg\Bigg)^2
            \notag \\
    = \wt z_0^{-1} E^{\wt Q} \, \frac{\wt z_0^2}{\wt z_T^2}\,
        \Bigg( H-x-\int_0^T \tht_t'\sg_t \,d M_t^0\Bigg)^2 \notag \\
    =\wt z_0^{-1} E^{\wt Q} \Bigg( \frac{H\wt z_0}{\wt z_T} -x-
        \int_0^T \psi_t^0(\sg)\, d\,\frac{\wt z_0^2}{\wt z_t^2} -
        \int_0^T \left( \psi_t^1(\sg)\right)' \,
       d\,\frac{M_t^0}{\wt z_t}\,\wt z_0\Bigg)^2 \notag \\
    := \ol{\cJ}(\sg,\psi^0,\psi^1) \;\;\;(\text{or} \;\; \ol{\cJ}(\sg,\psi) \;\;
        \text{with} \;\; \psi=(\psi^0,\psi^1)'\,), \label{4-3.28}
\end{gather}
where 
%   (3.29)
\begin{equation}\label{4-3.29}
\begin{gathered}
    \psi_t^1=\psi_t^1(\sg)=\sg_t'\tht_t, \\
        \psi_t^0=\psi_t^0(\sg)=\int_0^t \tht_s'\sg_s d M_s^0-
            \tht_t'\sg_t M_t^0, \;\;\; 0\leq t\leq T.
\end{gathered}
\end{equation}
Thus 
$$  
    \psi_t^1(\sg) =\psi_t^1(\sg^0)+\dl\psi_t^1(h), \quad 
    \psi_t^0(\sg) =\psi_t^0(\sg^0)+\dl\psi_t^0(h).
$$

Let (following RSch \cite{34}) 
%   (3.30)
\begin{equation}\label{4-3.30}
    \frac{H}{\wt z_T}\,\wt z_0 =E\left( \frac{H}{\wt z_T}\, \wt z_0\right) +
        \int_0^T (\psi_t^H)'dU_t+L_T,
\end{equation}
be the Galtchouk--Kunita--Watanabe decomposition of r.v. $\frac{H}{\wt z_T}\,\wt z_0$ w.r.t. $R^{(d+1)}$-valued $\wt Q$-local martingale $U=\big( \frac{\wt z_0}{\wt z}, \frac{M^0}{\wt z}\,\wt z_0\big)'$, where $\psi^H=(\psi^{0,H},\psi^{1,H})'\in L^2(U,\wt Q)$, the space of $F$-predictable processes $\psi$ such that $\int \psi' dU\in \cM^2(\wt Q)$ of martingale, and $L\in \cM_{0,\loc}^2(\wt Q)$, $L$ is $\wt Q$-strongly orthogonal to $U$. 

Remember that 
%   (3.31)
\begin{equation}\label{4-3.31}
    \psi=(\psi^0,\psi^1)'.
\end{equation}
Then, using (\ref{4-3.28}), (\ref{4-3.29}) and (\ref{4-3.30}) we can write for each $h$ 
%   (3.32)
\begin{gather}
    \cJ(\sg^0+\dl h,\psi) =\cJ(\sg^0,\psi)+\dl\cdot 2\wt z_0^{-1} \notag \\
    \times E^{\wt Q} \Bigg\{ \bigg[ 
        \bigg( x-E^{\wt Q} \,\frac{H}{\wt z_T}\,\wt z_0\bigg) -L_T +
        \int_0^T (\psi_t(\sg^0)-\psi_t^H)'dU_t\bigg]
        \int_0^T (\ol{\psi}_t(h))' dU_t\Bigg\} \notag \\
    +\dl^2 \wt z_0^{-1} E^{\wt Q} \Bigg[ \int_0^T (\ol{\psi}_t(h))' dU_t\Bigg]^2\notag\\
    =\cJ(\sg^0,\psi)+\dl\cdot 2\wt z_0^{-1} E^{\wt Q} 
        \Bigg[ \int_0^T (\psi_t(\sg^0)-\psi_t^H)' dU_t 
            \int_0^T (\ol{\psi}_t(h))' dU_t\Bigg] \notag \\
    +\dl^2 \wt z_0^{-1} E^{\wt Q} \Bigg[ \int_0^T (\ol{\psi}_t(h))' dU_t\Bigg]^2.
        \label{4-3.32}
\end{gather}
Using Proposition 8 of RSch \cite{34} we have for each $h$ 
$$
    \frac{\wt z_0}{\wt z_T}\,G_r(h,\Tht)=\Bigg\{ \int_0^T ({\psi}(h))' dU_t:
        {\psi}(h) \in L^2(U,\wt Q)\Bigg\},
$$
and hence by (\ref{4-3.22})
%   (3.33)
\begin{gather}
    E^{\wt Q} \Bigg(\int_0^T ({\psi}_t(h))' dU_t\Bigg)^2 \notag \\
    =E^{\wt Q}\,\frac{\wt z_0^2}{\wt z_T^2}\, G_T^2(h,\tht)= \wt z_0 EG_T^2(h,\tht) =
        \wt z_0 E\Bigg( \int_0^T \tht_t'\, dR_t^h\Bigg)^2 \notag \\
    =\wt z_0 E\Bigg( \int_0^T \tht_t' h_t d M_t^0\Bigg)^2 =
        \wt z_0 E \Bigg( \int_0^T \tht_t' h_t d \la M\ra_tk_t +
        \int_0^T \tht_t' h_t d M_t\Bigg)^2 \notag \\
    \leq const \Bigg[ E \Bigg( \int_0^T |\tht_t' h_t d \la M\ra_t k_t| \Bigg)^2 +
        E\Bigg( \int_0^T \tht_t' h_t d M_t\Bigg)^2 \Bigg] \notag \\
    \leq const \;r^2 E\int_0^T |\tht_t|^2 d C_t<\infty. \label{4-3.33}
\end{gather}

Further,
%   (3.34)
\begin{gather}
    \Bigg( E^{\wt Q} \Bigg[ \int_0^T (\psi_t(\sg^0)-\psi_t^H)' dU_t 
        \int_0^T (\psi_t(h))' dU_t \Bigg]\Bigg)^2 \notag \\
    \leq E^{\wt Q} \Bigg( \int_0^T (\psi_t(\sg^0)-\psi_t^H)' dU_t \Bigg)^2 
        E^{\wt Q} \Bigg( \int_0^T (\psi_t(h))' dU_t \Bigg)^2<\infty.\label{4-3.34}
\end{gather}
From these estimates we conclude that: 
%   (3.35)
\begin{equation}\label{4-3.35}
    1) \;\;\; D\ol{\cJ} (\sg^0,h,\psi)=
        2\wt z_0^{-1} E^{\wt Q} \int_0^T (\psi_t(\sg^0)-\psi_t^H)' 
            d\la U\ra_t \psi_t(h)<\infty,\;\;
\end{equation}
thanks to (\ref{4-3.33}).

2) $D\ol{\cJ}(\sg^0,h,\psi)|_{h\equiv 0} =0$, since $\psi(0)=0$ by (\ref{4-3.31}) and (\ref{4-3.29}).

Thus 
%   (3.36)
\begin{equation}\label{4-3.36}
    \sup_{h\in \cH} D \ol{\cJ}(\sg^0,h,\psi)\geq 0.
\end{equation}

3) From (\ref{4-3.34}) and (\ref{4-3.33}) we get 
\begin{gather*}
    (D\ol{\cJ}(\sg^0,h,\psi))^2 \leq const \;\wt z_0^{-2} r^2 \\
    \times E^{\wt Q} \int_0^T (\psi_t(\sg^0)-\psi_t^H)' d\la U\ra_t 
        (\psi_t(\sg^0)-\psi_t^H) 
        E\int_0^T |\tht_t|^2 d C_t<\infty.
\end{gather*}
Thus $|D\ol{\cJ}(\sg^0,h,\psi)|$ is estimated by the expression which does not depend on $h$, and is equal to zero if we substitute $\psi_t(\sg^0)\equiv \psi_t^H$, $0\leq t\leq T$.

Hence, by (\ref{4-3.36})
%   (3.37)
\begin{equation}\label{4-3.37}
    0\leq \sup_{h\in \cH} D\ol{\cJ}(\sg^0,h,\psi)|_{\psi\equiv \psi^H} \leq 
        \sup_{h\in \cH} |D\ol{\cJ}(\sg^0,h,\psi)|\big|_{\psi\equiv \psi^H}=0
\end{equation}

Further, from (\ref{4-3.36}) follows that we can take $c\in [0,\infty)$ in (\ref{4-3.26}).

Now substituting $\psi\equiv \psi^H$ into $\ol{\cJ}(\sg^0,\psi)$ and $D\ol{\cJ}(\sg^0,h,\psi)$ we get 
$$  
    \ol{\cJ}(\sg^0,\psi^H) =\min_\psi \ol{\cJ}(\sg^0,\psi) =\wt z_0^{-1} (E^{\wt P} H-x)^2+
        \wt z_0^{-1} E^{\wt Q} L_T^2
$$
(see Lemma 5.1 of GLP \cite{11}) and 
$$  
    \sup_{h\in \cH} \frac{D \ol{\cJ}(\sg^0,h,\psi^H)}{\ol{\cJ}(\sg^0,\psi^H)}=0.
$$

Hence the constraint of problem (\ref{4-3.26}) is satisfied.

%   Remark 3.4
\begin{remark}\label{r4-3.4}
If $x=E^{\wt P}H$ and $L_T\equiv 0$, then we get 
$$  
    \frac{D \ol{\cJ}(\sg^0,h,\psi^H)}{\ol{\cJ}(\sg^0,\psi^H)}=\frac{0}{0}
$$
which is assumed to be zero, since if we consider the shifted risk functional $\wt\cJ=\ol{\cJ}+1$, the optimization problem and the optimal trading strategy will not change, but $D\wt\cJ(\sg^0,h,\psi^H)=D\ol{\cJ}(\sg^0,h,\psi^H)=0$ and $\wt\cJ(\sg^0,\psi^H)=1$.
\end{remark}

Finally, using Proposition 8 of RSch \cite{34} we arrive at the following

%   Theorem 3.1
\begin{theorem}\label{t4-3.1}
In Model $(\ref{4-3.10})$ under conditions {\rm (c.1)} and {\rm (c.2)} the optimal mean-variance robust trading strategy $($in the sense of Definition $\ref{d4-3.1})$ is given by the formula 
%   (3.38)
\begin{equation}\label{4-3.38}
    \tht_t^* =((\sg_t^0)')^{-1} [\psi_t^{1,H} +\zt_t(V_t^*-(\psi_t^H)'U_t)],
        \quad 0\leq t\leq T,
\end{equation}
where
\begin{gather*}
    \psi_t^H=(\psi_t^{0,H},\psi_t^{1,H})', \quad 
        U_t=\left( \frac{\wt z_0}{\wt z_t}\,,\frac{M_t^0}{\wt z_t}\,\wt z_0 \right)', \\
    V_t^*=\frac{\wt z_0}{\wt z_t} \Bigg( x+\int_0^t (\psi_t^H)' dU_t\Bigg),
\end{gather*}
$\psi_t^H$ and $\zt_t$ are given by the relations $(\ref{4-3.30})$ and $(\ref{4-3.11})$, respectively, $\wt z_t$ is defined in $(\ref{4-3.11})$.
\end{theorem}

\end{document}